\documentclass[journal]{IEEEtran}

\usepackage{amssymb}
\usepackage{amsmath}
\usepackage{cases}
\usepackage{amsthm}
\usepackage{enumerate}
\usepackage{varioref}
\usepackage{psfrag,balance}
\usepackage{graphics}
\usepackage{psfrag}
\usepackage{overpic}
\usepackage{extarrows}
\usepackage{cite}
\usepackage{leftidx}
\usepackage{color}
\usepackage{cite}
\usepackage{slashbox}

\graphicspath{{Image/}}

\newcommand{\diag}{\mathrm{diag}}
\newcommand{\D}{\mathfrak{D}}
\newcommand{\E}{\mathrm{E}}

\newcommand{\dis}{\mathrm{dis}}

\newcommand{\degree}{\mathrm{deg}}
\theoremstyle{plain}
\newtheorem{lemma}{Lemma}
\newtheorem{coro}{Corollary}
\newtheorem{thm}{Theorem}

\theoremstyle{definition}
\newtheorem{definition}{Definition}

\theoremstyle{remark}
\newtheorem{remark}{Remark}
\ifCLASSINFOpdf

\else

\fi

\begin{document}

\title{Containment Control of Multi-Agent Systems with Dynamic Leaders Based on a $PI^n$-Type Approach}
%\title{Polynomial Trajectory Containment Control of Multi-Agent System}
\author{Long Cheng, Yunpeng Wang, Wei Ren, Zeng-Guang Hou, and Min Tan
\thanks{Long Cheng, Yunpeng Wang, Zeng-Guang Hou and Min Tan are with the State Key Laboratory of Management and Control for Complex Systems, Institute of Automation, Chinese Academy of Sciences, Beijing 100190, China.}
\thanks{Wei Ren is with the Department of Electrical and Computer Engineering, University of California at Riverside, CA 92521, USA.}
\thanks{Please address all correspondences to Dr. Long Cheng at Email: chenglong@compsys.ia.ac.cn; Tel: 8610-82544522; Fax: 8610-82544794.}}

\maketitle

\begin{abstract}
  This paper studies the containment control of multi-agent systems with multiple dynamic leaders in {both the continuous-time domain and the discrete-time domain}.
  The leaders' motions are described by the $n$th-order polynomial trajectories. This setting makes practical sense because given some critical points, the leaders' trajectories are usually planned by the polynomial interpolations.
  In order to drive all followers into the convex hull spanned by the leaders, a $PI^n$-type containment algorithm is proposed ($P$ and $I$ are short for {\it Proportional} and {\it  Integral} , respectively; $I^n$ implies that the algorithm includes {up to the $n$th-order} integral terms).
  It is theoretically proved that the $PI^n$-type containment algorithm is able to solve the containment problem of multi-agent systems where the followers are described by any order integral dynamics. Compared to the previous results on the multi-agent systems with dynamic leaders, the distinguished features of this paper are that: (1) the containment problem is studied not only in the continuous-time domain but also in the discrete-time domain while most existing results only work in the continuous-time domain; (2) to deal with the leaders with the $n$th-order polynomial trajectories, existing results require the follower's dynamics to be the $(n+1)$th-order integral while the followers considered in this paper can be described by any-order integral dynamics; (3) the ``sign'' function is not employed in the proposed algorithm, which avoids the chattering phenomenon; and {(4) both the disturbance and the measurement noise are taken into account.} Finally, some simulation examples are given to demonstrate the effectiveness of the proposed algorithm.
  %Furthermore, in order to illustrate the practical value of the proposed approach, an application, the containment control of multiple mobile robots is studied.
  %Finally, two simulation examples are given to demonstrate the effectiveness of the proposed algorithm.
\end{abstract}


\IEEEpeerreviewmaketitle

\section{Introduction}

Recently, the distributed coordinated control of multi-agent systems (MASs) has become a research focal in the systems and control community. Roughly speaking, agents in concern can be divided into two categories: leaders and followers. Depending on whether there are leaders in MASs, the coordinated control problem becomes the consensus problem (leaderless case) \cite{Hu15TCYB,Li15TCYB,Hou09SMCB}; the leader-following problem (single leader case) \cite{Cheng10TNN}; and the containment problem (multiple leaders case).
This paper mainly focuses on the containment problem of MASs because, from one side, the containment problem roots in some natural phenomena such as the relationship between sheepdogs and sheep \cite{Vaugha00RAS} and the relationship between female silkworm moths and male silkworm moths \cite{Haque08ACC};
from the other side, the containment problem has many practical applications such as the mixed containment-sensing problem \cite{Galbusera13SCL} and the coordinated control of a group of mobile robots  \cite{Parker03RS,Cao11TCST,Wang15CCC}.

Looking back at the history of the containment problem, the rapid development started after the publication of \cite{Dimarogonas06CDC,Ferrari06WHS}.
A leader-based containment control strategy for multiple unicycle agents was introduced in \cite{Dimarogonas06CDC}, where the containment problem was interpreted as a combination of the formation and agreement control problems.
The leaders were convergent to a desired formation while the followers converged to the convex hull spanned by the leaders.
A similar containment problem of MASs with single-integrator dynamics was studied in \cite{Ji08TAC}, where consensus-like interaction rules were designed for the followers while a hybrid ``Stop-Go'' policy was applied to the leaders.
Since then, a great number of results concerning the containment control have been reported.
According to the type of the agent's dynamics, these results can be divided into four categories: (1) single-integrator dynamics \cite{Notarstefano11Automatica,Cao12Automatica,Tang12AAA}; (2) double-integrator dynamics \cite{Wang14TCYB,Liu12Automatica,Li14NAHS,Cao11TCST,Li12TAC,Zhang14SCL,Lou12Automatica,Wang14Automatica}; (3) general linear dynamics \cite{Liu13SCL,Ma14Neuro,Li13JRNC,Liu15Automatica}; (4) Euler-Lagrange dynamics \cite{Mei12Automatica,Yoo14ES,Meng10Automatica}; and (5) nonlinear dynamics \cite{XKWang14Cyb}.

In \cite{Notarstefano11Automatica}, the containment problem of MASs with undirected switching communication topologies was studied.
However, in practice, the communication link is usually a one-way channel.
For this reason, the containment problem of MASs with directed communication topologies has been widely investigated recently.
In \cite{Cao12Automatica}, it was shown that the necessary and sufficient condition for achieving the containment of single-integrator MASs with a directed topology was that for each follower, there existed at least one leader that had a directed path to this follower.
This condition was also proved to be necessary and sufficient for the containment problem of double-integrator MASs in \cite{Liu12Automatica,Li14NAHS}.
Experimental validations on a team of mobile robots were conducted in \cite{Cao11TCST}. The finite-time containment problem of double-integrator MASs was investigated in \cite{Wang14TCYB}.
In \cite{Li12TAC,Zhang14SCL}, containment control algorithms were proposed for double-integrator MASs based on only position measurements.
The containment problem of double-integrator MASs with randomly switching topologies was investigated in \cite{Lou12Automatica}, where the switching signal was described by a continuous-time irreducible Markov chain.
It was proved that the containment problems could be solved if and only if for each follower, there existed at least one leader which had a directed path to this follower in the union graph of all possible communication graphs.
Because the communication noise is unavoidable in practical applications, the noise effect in containment problems of single-integrator and double-integrator MASs were studied in \cite{Tang12AAA} and \cite{Wang14Automatica}, respectively.
The results in \cite{Liu13SCL} show that the containment control of general linear MASs can be achieved by applying a state-feedback algorithm.
When the agents' states were unavailable, output-feedback based containment control algorithms were proposed for general linear MASs in \cite{Ma14Neuro,Li13JRNC}. In \cite{Liu15Automatica}, the communication constraint such as the non-uniform delay was considered in studying the containment problem. 
For Euler-Lagrange MASs with uncertainties, adaptive containment algorithms were proposed based on sliding-mode estimators and neural networks in \cite{Mei12Automatica} and \cite{Yoo14ES}, respectively.
Furthermore, the finite-time containment problem of Euler-Lagrange MASs was studied in \cite{Meng10Automatica}. Finally, in \cite{XKWang14Cyb}, the containment problem of the second-order locally Lipschitz nonlinear MASs was solved within the framework of the nonlinear input-to-state stability.

Although great effort has been made to address various factors in the containment control of MASs, there are still some limitations in the existing results.
Let us first consider an application scenario shown in Fig. \ref{fig:back}.
In this application, a group of mobile robots are required to move across a partially unknown area through a narrow safe tunnel.
There are two kinds of robots: the master robots and the slave robots.
The master robots are capable of self-navigation, while the slave robots can only measure the relative positions with its neighbor robots.
This task can be solved by the containment control strategy:
\begin{enumerate}
	\item Master robots act as leaders. For each master robot, design a reference trajectory which is inside the safe channel. Let each master robot move along its corresponding reference trajectory.
	\item Slave robots act as followers. Let all slave robots move into the area surrounded by master robots and move together with master robots.
\end{enumerate}
Then the challenge is how to design the reference trajectories for master robots.
For each master robot, we can select a sequence of suitable reference points inside the safe tunnel.
A polynomial trajectory, which goes through these points, can be constructed by the polynomial interpolation. The obtained trajectory is checked whether it is inside the safe tunnel. If not (caused by the Runge's phenomenon), we need to re-select the reference points and construct the new polynomial trajectory. This process is repeated until the trajectory is within the safe tunnel. Then the following $n$th-order polynomial can be determined
\begin{equation}\label{eq:intro_1}
x(t)=a_0+a_1t+\cdots+a_{n}t^{n},
\end{equation}
whose trajectory goes through the selected $(n+1)$ reference points ($a_0,\cdots,a_n$ are coefficients determined by these reference points).
\begin{figure}
	\centering
	\includegraphics[width=0.9\hsize]{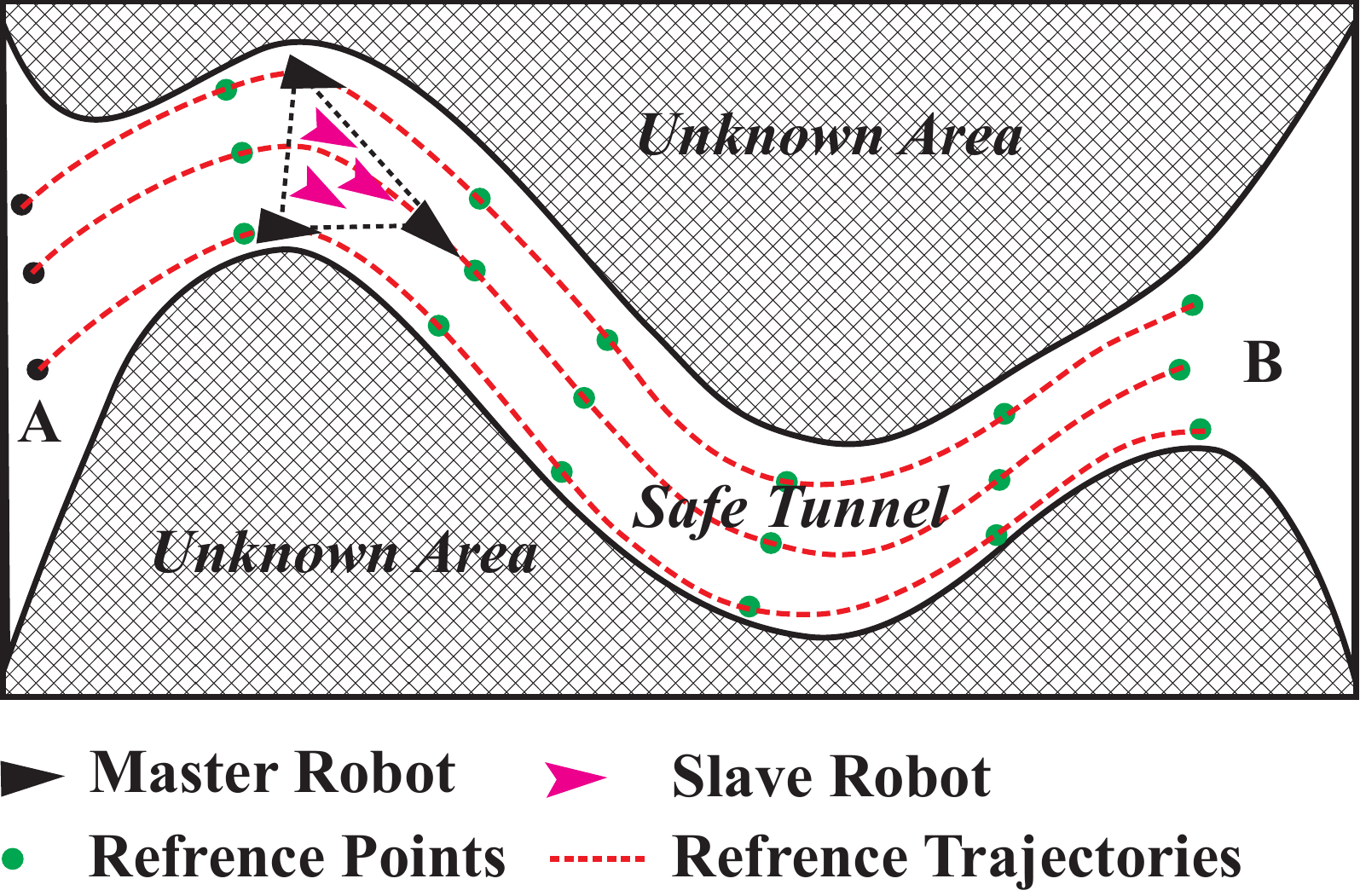}\\
	\caption{A group of mobile robots travel across a partially unknown area.}\label{fig:back}
\end{figure}
The polynomial trajectory can be used as the reference trajectory for the master robot.

\color{black}
In \cite{Notarstefano11Automatica,Cao12Automatica,Tang12AAA,Liu12Automatica,Li14NAHS,Cao11TCST,Li12TAC,Zhang14SCL,Lou12Automatica,Wang14Automatica,Mei12Automatica,Meng10Automatica}, every leader has the following single-integrator or double-integrator dynamics:
\begin{equation}\label{eq:intro_2}
\dot x(t)=u(t)\ \quad \text{or}\quad \
\begin{cases}
\dot x(t)=v(t)\\ \dot v(t)=u(t)
\end{cases},
\end{equation}
where the control input $u(t)$ is assumed to be bounded or even zero.
Obviously,  {if the control input of \eqref{eq:intro_2} is bounded, then this controller cannot generate the polynomial trajectory defined by \eqref{eq:intro_1} with the order $n\ge 3$}.
If the leader has the $n$th-order integrator or the $n$-dimensional linear dynamics \cite{Liu13SCL,Ma14Neuro,Li13JRNC}, then the polynomial trajectory can be generated by properly selecting the parameters in the system matrix of the leader's dynamics.
However, by doing so, it is required that all followers have the same dynamics as the leader \cite{Liu13SCL,Ma14Neuro,Li13JRNC}.
In reality, the follower's dynamics has no relationship with the order of the leader's polynomial trajectory.
In fact, under the $n$th-order polynomial trajectories for leaders, how to solve the containment problem of MASs with single-integrator followers is not answered yet. In addition, from the controller design point of view, most containment algorithms include the non-smooth ``sign'' function to deal with the dynamic leaders \cite{Cao12Automatica,Cao11TCST,Li12TAC,Zhang14SCL,Mei12Automatica,Meng10Automatica}, which would cause the harmful ``chattering'' phenomenon.
Furthermore, all results aforementioned above only study the containment problem in the continuous-time domain.
The counterpart results in the discrete-time domain are not clear.

Inspired by the above observations, this paper investigates the containment problem of MASs with dynamic leaders in both the  {continuous-time domain and the discrete-time domain}.
It is assumed that each leader's motion is described by a corresponding polynomial trajectory.
%Besides, any continuous function defined on a closed and bounded interval can be uniformly approximated by a polynomial function to any degree of accuracy (Weierstrass approximation Theorem).
Every follower is first assumed to have the single-integrator dynamics.
A so-called $PI^n$-type ($P$ and $I$ are short for {\it Proportional} and {\it  Integral}, respectively) algorithm is proposed to solve the containment problem, where $I^n$ implies that the algorithm includes  {up-to the $n$th-order} integral terms.
It turns out that the $PI^n$-type algorithm is able to solve the containment problem if for each follower, there exists at least one leader which has a directed path to this follower.
Then the obtained results are extended to the case where the followers are described by the high-order integral dynamics.
In this case, the  $PI^n$-type algorithm is modified as a $PI^{n-m}D^m$-type algorithm ($D$ is short for  {\it derivative}; and $D^m$ implies that the algorithm includes {up to the $m$th-order} differential terms).
 {Then, the counterpart results in the discrete-time domain are presented. Moreover, effects of the disturbance and the measurement noise are also taken into account in this paper.}
%A practical application of the proposed $PI^n$-type algorithm is given: the coordinated control of multiple mobile robots.
 {Compared to the previous results, the contributions of this paper can be summarized as follows:
\begin{enumerate}
	\item The proposed algorithms can solve the containment problem with dynamic leaders in both the continuous-time domain and the discrete-time domain.
	\item The follower can be described by any-order integral dynamics.
	\item There is no discontinuous ``sign'' function in the proposed controllers, which avoids the ``chattering'' phenomenon.
	\item Effects of the disturbance and the measurement noise are taken into account.
%	The containment problem of MASs with disturbance and measurement noise
%	For the discrete-time case, the algorithm can also solve the containment problem of multi-agent system with measurement noise in stochastic sense.
\end{enumerate}}
It is noted that there are some recent results on the ``$PI$''-type consensus algorithms for MASs \cite{Andreasson14TAC,Burbano15TAC}. This kind of consensus algorithms has also been applied in the distributed filter of distributed parameter systems \cite{Demetriou14SCL}.
Compared to \cite{Andreasson14TAC,Burbano15TAC}, the distinguished features of this paper mainly lie in the following five aspects:
\begin{enumerate}
	\item The control objective of \cite{Andreasson14TAC,Burbano15TAC} is to solve the leaderless consensus problem of multi-agent systems with disturbances. The aim of this paper is to solve the containment problem with dynamic leaders.
	
	\item In \cite{Andreasson14TAC,Burbano15TAC}, the agent is described by the single-integrator or double-integrator dynamics. In this paper, the leaders move along the polynomial trajectories, and the followers can be described by any-order integral dynamics.
	
	\item In \cite{Andreasson14TAC,Burbano15TAC}, the main role of the integral term is to attenuate constant disturbances.
	In this paper, the main purpose of employing the integral-terms is to eliminate the containment error caused by the polynomial trajectory.
	
	\item The algorithms proposed in \cite{Andreasson14TAC,Burbano15TAC} can only deal with constant disturbances, while the ones proposed in this paper can attenuate some kinds of time-varying disturbances (polynomial-type disturbance).
	
	\item In \cite{Andreasson14TAC,Burbano15TAC}, only results in the continuous-time domain are presented. This paper studies the containment problem of MASs in both the continuous-time domain and the discrete-time domain.
\end{enumerate}

\color{black}
The remainder of this paper is organized as follows: Section \ref{sec:pre} gives some preliminary results on the containment problem with dynamic leaders;  {Section \ref{sec:dis} presents a containment algorithm  and the related theoretical analysis in the continuous-time domain, where the followers are described by the single-integrator dynamics};
Section \ref{sec:highorder} discusses how to generalize the obtained results to the case where the followers are described by the high-order integral dynamics; counterpart results in the discrete-time domain are presented in Section \ref{sec:ext}; Section \ref{sec:conclude} concludes this paper with final remarks.

\noindent \textbf{Notations}:  {$\textbf{1}_n = (1,\cdots,1)^T \in \mathbb{R}^n$; $\textbf{0}_n = (0,\cdots,0)^T \in
\mathbb{R}^n$}; $I_n$ denotes the $n\times n$ dimensional identity
matrix;  {$\textbf{0}_{m\times n}\in {\mathbb R}^{m\times n}$ denotes the $m\times n$ dimensional zero matrix};
$\otimes$ denotes the Kronecker product.
$\mathbb{N}$ and ${\mathbb N}^+$ denote the set of natural numbers and the set of positive natural numbers, respectively. For a given vector $p \in \mathbb{R}^n$ and a set $\Omega \subseteq \mathbb{R}^n$, the distance between $p$ and $\Omega$ is defined as $\dis(p,\Omega)=\inf_{y\in\Omega}\|p-y\|_2$.
For a given matrix $X$, $\|X\|_2$ denotes its 2-norm;  {$\|X\|_F$ denotes its Frobenius norm; $X^T$ denotes its transpose; and $X^H$ denotes its conjugate transpose}. diag$(\cdot)$ denotes a block diagonal matrix formed by its inputs.
For a complex number $c$, $\Re(c)$ denotes its real part.
 {For a given random variable or vector $x$, $\E(x)$ denotes its mathematical expectation.} %$\D$ denotes the differential operator or the difference operator.
 %For a set $S=\{x_1,\cdots,x_n\}$ the convex hull spanned by $S$ is denoted by $\co(S)=\{\sum_{i=1}^n\alpha_ix_i|\alpha_i\in{\mathbb R},\alpha_i\ge0,\sum_{i=1}^n\alpha_i=1\}$.

\section{Preliminaries \& Problem Formulation}\label{sec:pre}
Consider a MAS composed of $M+N$ agents.
 {Define two sets ${\cal M}=\{1,2,\cdots,M\}$ and ${\cal N}=\{M+1,M+2,\cdots,M+N\}$. Motivated by \cite{Cao12Automatica,Tang12AAA,Liu12Automatica,Li14NAHS,Lou12Automatica,Wang14Automatica,Liu13SCL,Wang14Cyb,Meng11SMCB}, the interaction topology of the MAS is modeled by a weighted digraph $\cal G=\{V_G,E_G,A_G\}$, where ${\cal V_G}=\{v_1,v_2,\cdots,v_{M+N}\}$, ${\cal E_G}\subset{\cal V_G}\times{\cal V_G}=\{\epsilon_{ij}|i,j\in {\cal M\cup N}\}$ and ${\cal A_G}=[\alpha_{ij}]\in{\mathbb R}^{(M+N)\times (M+N)}$ are the node set, the directed edge set and the adjacency matrix, respectively.}
Node $v_i$ denotes agent $i$; $\epsilon_{ij}\in{\cal E_G}$ means that there is an information flow from agent $i$ to agent $j$; $\alpha_{ji}$ denotes the weight associated with the directed edge $\epsilon_{ij}$.
The element of $\cal A_G$ satisfies that $\alpha_{ji}>0 \Leftrightarrow \epsilon_{ij}\in{\cal E_G}$  and  $\alpha_{ji}=0 \Leftrightarrow \epsilon_{ij}\notin{\cal E_G}$.
It is assumed that there is no self-loop in  ${\cal V_G}$ ($\epsilon_{ii}\notin {\cal E_G}$ and $\alpha_{ii}=0$, $i\in{\cal M \cup \cal N}$).
If $\epsilon_{ij}\in{\cal E_G}$, then agent $i$ is called the parent of agent $j$.
The neighborhood of node $v_i$ is defined as ${\cal N}_i=\{v_j|\epsilon_{ji}\in{\cal E_G}\}$.
The in-degree of node $v_i$ is defined as $\degree_{in}(v_i)=\sum_{j\in{\cal N}_i}\alpha_{ij}$.
The Laplacian matrix of $\cal G$ is defined as $\cal L_G=D_G-A_G$, where ${\cal D_G}=\diag(\deg_{in}(v_1),\deg_{in}(v_2),\cdots,\deg_{in}(v_{M+N}))$.
A directed path from node $v_{i_1}$ to node $v_{i_n}$ is a sequence of end-to-end directed edge $\epsilon_{i_1i_2}, \epsilon_{i_2i_3},\cdots,\epsilon_{i_{n-1}i_n}$ where $\epsilon_{i_ji_{j+1}}\in{\cal V_G}\ (j=1,\cdots,n-1)$.

In this paper,  {an agent is called a leader if it has no parent; otherwise it is called a follower.}
Without loss of generality, it is assumed that the agents labeled from $1$ to $M$ are the leaders while the agents labeled from $M+1$ to $M+N$ are the followers.
Hence the Laplacian matrix of the interaction topology graph has the following form
\begin{equation}\label{eq:4}
  {\cal L_G}=
  \begin{bmatrix}
    \textbf{0}_{M\times M} & \textbf{0}_{M\times N}\\
    L_1 & L_2
  \end{bmatrix},
\end{equation}
where $L_1\in{\mathbb R}^{N\times M}$ and $L_2\in{\mathbb R}^{N\times N}$.

Throughout this paper, it is assumed that the following two assumptions hold.
\begin{description}
	\item[(A1)] For each follower, there exists at least one leader that has a directed path to this follower.
	\item[(A2)] Each follower can only measure the relative positions between itself and its neighbors.
\end{description}
\begin{lemma}[\cite{Meng10Automatica}]\label{lem:Laplacian}
  Under Assumption (A1),
  \begin{itemize}
  	\item  all eigenvalues of $L_2$ defined in \eqref{eq:4} have positive real parts;
  	\item  each entry of $-L_2^{-1}L_1$ is nonnegative and the row sum of $-L_2^{-1}L_1$ equals to one.
  \end{itemize}
\end{lemma}

The control objective is to design the control algorithms for the followers such that all followers are convergent into the convex hull spanned by the leaders (containment problem) while the leaders move along some predesigned trajectories. To this end, how to describe the motions of the leaders and followers should be given.
Let us first consider the  {continuous-time} domain case. The position of agent $i$ at time  $t$ is denoted by  {$x_i(t)\in{\mathbb R}^p$} (it is assumed that the agent moves in the $p$-dimensional space).
The $i$th leader's motion is assumed to move along the following  {$n$th-order} polynomial trajectory
\begin{equation}\label{eq:16}
 {x_i(t)={\bf a}^i_0+{\bf a}^i_1t+\cdots+{\bf a}^i_{n}t^{n}, \quad i\in{\cal M},}
\end{equation}
where ${\bf a}_j^i\in{\mathbb R^p}, \;j=0,\cdots,n$. The reason of employing the polynomial trajectory is that the robot's trajectory is usually planned by the polynomial interpolation.
%These coefficients can be determined according to specific tasks.
%For example, when planning the trajectory of robots, some critical points in the operation space are first fixed.
%Then each robot is required to pass through these points along a smooth trajectory. This problem is usually solved by the polynomial interpolation method, which results in the polynomial trajectory defined by \eqref{eq:16}.
In this paper, we do not care about the dynamics of the leaders.
The leaders can be considered as the reference signals.
The motion of the $(i-M)$-th follower ($i\in{\cal N}$) is described by the following first-order differential equation
\begin{equation}\label{eq:1}
 {\D x_i(t)=u_i(t)+\delta_i(t)},
\end{equation}
 {where $u_i(t)\in{\mathbb R}^p$ and $x_i(t)\in{\mathbb R}^p$.
In \eqref{eq:1}, the symbol $\D$ denotes the differential operator (namely $\D x_i(t)=\dot x_i(t)$ and $\D^{n}x(t)=\D(\D^{n-1}x(t))=x^{(n)}(t)$).
And $\delta_i(t)={\bf b}_0^i+{\bf b}_1^it+\cdots+{\bf b}_r^it^r$ is the polynomial-type disturbance, where $r\in{\mathbb N}$ and ${\bf b}_j^i\in{\mathbb R}^p\;(j=0,\cdots,r)$}.
The inversion of this  {differential} operator $\D$ is defined as  {$\D^{-1} x(t)=\int_{0}^{t}x(s)ds$}.
It is easy to see that $\D(\D^{-1}x(t))=x(t)$.

By the above terminologies, the containment problem can be formally defined as follows.
\begin{definition}\label{def:1}
	The containment problem of MASs is solved if all followers' positions are convergent into the convex hull spanned by the leaders' positions. That is
	\begin{equation}
	 {\lim_{t\to\infty}\dis(x_i(t),co_L(t))=0,\ i\in{\cal N},\nonumber}
	\end{equation}
	where $co_L(t)=co\{x_1(t),\cdots,x_M(t)\}=\{\sum_{i\in{\cal M}}\mu_ix_i(t)|\sum_{i\in{\cal M}}\mu_i=1,\ \mu_i\ge0\}$ is the convex hull spanned by the leaders' positions at time $t$.
\end{definition}
\
\begin{definition}
	The containment error of MASs as time $t$ is defined as $E_r(t)=\sum_{i\in{\cal N}}\dis(x_i(t),co_L(t))$.
\end{definition}
It is easy to see that the containment problem of MASs is solved if and only if $\lim_{t\to\infty}E_r(t)=0$.

\section{ {Containment Control of MASs in Continuous-Time Domain}}\label{sec:dis}
Let $e_{ji}(t) = x_j(t) - x_i(t)$ denote the relative state between agent $j$ and agent $i$.
The following containment controller is proposed for the $i$th agent
\begin{equation}\label{eq:11}
u_i(t)=\sum_{l=0}^{n}\kappa_l\D^{-l}\left(\sum_{j\in{\cal M}}\alpha_{ij}e_{ji}(t)+\sum_{j\in{\cal N}}\alpha_{ij}e_{ji}(t)\right),
\end{equation}
where $\{\kappa_l;\ l=0,\cdots,n\}$ are parameters to be determined.
Because \eqref{eq:11} includes the proportional term $\sum_{j\in{\cal M\cup N}}\alpha_{ij}e_{ji}(t)$ and up to the $n$th-order integral terms \{$\D^{-l}\sum_{j\in{\cal M\cup N}}\alpha_{ij}e_{ji}(t)$, $l=1,\cdots,n$\}, \eqref{eq:11} is called $PI^n$-type algorithm.
The main purpose of employing the integral terms is to eliminate the containment error caused by the polynomial trajectory.

\color{black}
Let $\xi_i(t)=(x_i^T(t),\D x_i^T(t),\cdots,{\D^{n}x_i}^T(t))^T$.
Then the $i$th agent's dynamical behavior can be described by the following differential equation
 {
\begin{equation}\label{eq:3}
\begin{cases}
\D\xi_i(t)=(A\otimes I_p)\xi_i(t), & i\in{\cal M},\\
\D\xi_i(t)=(A\otimes I_p)\xi_i(t)+(B\otimes I_p) \bar u_i(t)\\
\quad\quad\quad\quad+\D^{n}\delta_i(t), & i\in{\cal N},
\end{cases}
\end{equation}}
where $\bar u_i(t)\triangleq\D^{n} u_i(t)=\sum_{j\in{\cal M\cup N}}K(\xi_j(t)-\xi_i(t)$, $K=(\kappa_n,\cdots,\kappa_0)$, $B=(0,\cdots,0,1)^T\in{\mathbb R}^{n+1}$, and
\begin{equation}
A=
\begin{bmatrix}
0 & 1 & \cdots & 0\\
\vdots & \vdots & \ddots & \vdots\\
0 & 0 & \cdots & 1\\
0 & 0 & \cdots & 0
\end{bmatrix}\in {\mathbb R}^{(n+1)\times (n+1)}.\nonumber
\end{equation}
Let $\Xi_L(t)=(\xi_1^T(t),\cdots,\xi_{M}^T(t))^T$ and $\Xi_F(t)=(\xi_{M+1}^T(t),$ $\cdots,\xi_{M+N}^T(t))^T$. Then the closed-loop dynamics of the MAS can be rewritten in the following compact form
 {
\begin{equation}
\begin{cases}
\D\Xi_L(t)=&\!\!\!(I_M\otimes A\otimes I_p)\Xi_L(t), \\
\D\Xi_F(t)=&\!\!\!\big(I_N\otimes A\otimes I_p-L_2\otimes BK\otimes I_p\big)\Xi_F(t)\\
&\!\!\!\!\!\!\!\!\!\!\!\!\!\!\!\!\!\!\!\!\!\!\!\!-\big(L_1\otimes BK\otimes I_p\big)\Xi_L(t)+(I_N\otimes B \otimes I_p)\Delta(t),
\end{cases}
\end{equation}
where $\Delta(t)=\diag(\D^n\delta_{M+1}^T(t),\cdots,\D^n\delta_{M+N}^T(t))$.}
This leads to that
 {
\begin{multline}\label{eq:9}
\D{\hat \Xi}_F(t)=(I_N\otimes  A\otimes I_p- L_2\otimes BK\otimes I_p)\hat\Xi_F(t)\\
+(I_N\otimes B \otimes I_p)\Delta(t),
\end{multline}}
where $\hat \Xi_F(t)=\Xi_F(t)+(L_2^{-1}L_1\otimes I_{(n+1)p}) \Xi_L(t)$.
By Lemma 1 and Definition 1, if ${\hat \Xi}_F(t)$ is convergent to zero, then the containment problem is solved.

 {
\begin{thm}\label{thm:con_1}
	Assume all leaders move along polynomial trajectories described by \eqref{eq:16}; and all followers have single-integrator dynamics described by \eqref{eq:1}.
	Let $P$ denote the positive definite solution to the following matrix inequality
	\begin{equation}\label{eq:43}
	A^TP+PA+I_{n+1}-PBB^TP\le 0.
	\end{equation}
	If the order of the polynomial disturbance $\delta_i(t)$ in \eqref{eq:1} is not greater than $n-1$ (namely, $r\le n-1$), then the containment problem of MASs can be solved by \eqref{eq:11} with $K=\varepsilon B^TP$ where $\varepsilon\ge 0.5\max\{1,\sigma_{\min}^{-1}\}$, $\sigma_{\min}\in(0,\lambda_{\min})$, and $\lambda_{\min}=\min\{\Re(\lambda_i)| \lambda_i\ ( i=1,\cdots,N) \text{ is the minimum eigenvalue of } L_2\}$.
\end{thm}}

\begin{proof}
	First, it is proved that $(I_N\otimes  A- L_2\otimes BK)\otimes I_p$ is a {\it Hurwitz} matrix.
	By Schur decomposition, there must exist a transformation matrix $T_c$ such that
	\begin{equation}
	\Lambda_c\triangleq T_c L_2T_c^{-1}=
	\begin{bmatrix}
	\lambda_1& * &\cdots&* \\
	0 &\lambda_2 & \cdots&*\\
	\vdots &\vdots&\ddots&\vdots\\
	0&0&\cdots&\lambda_N
	\end{bmatrix},\nonumber
	\end{equation}
	which leads to $(I_N\otimes  A- L_2\otimes BK)=(T_c\otimes I_n)(I_N\otimes  A- \Lambda_c\otimes BK)(T_c^{-1}\otimes I_n)$.
	Hence, the diagonal elements of $I_N\otimes  A- \Lambda_c\otimes BK$  are $A-\lambda_iBK$ ($i=1,\cdots,N$).
	
	By Lemma \ref{lem:Laplacian}, all eigenvalues $\{\lambda_i;i=1,\cdots,N\}$ have positive real parts. Hence,  {$(0,\lambda_{\min})$} is not empty.
	It is easy to see that $\varepsilon\Re(\lambda_i)\ge0.5\Re(\lambda_i)\max\{1,\sigma_{\min}^{-1}\}\ge 0.5$.
	This together with \eqref{eq:43} leads to that
	\begin{equation*}
	\begin{split}
	&(A-\lambda_iBK)^HP+P(A-\lambda_iBK)\\
	=&(A-\varepsilon\lambda_iBB^TP)^HP+P(A-\varepsilon\lambda_iBB^TP)\\
	=& A^TP+PA-2\varepsilon\Re(\lambda_i) PBB^TP\\
	\le&-I_{n+1}+(1-2\varepsilon\Re(\lambda_i))PBB^TP<0.
	\end{split}
	\end{equation*}
	By Lyapunov stability theory, $A-\lambda_iBK$ ($i=1,\cdots,N$) are all {\it Hurwitz} matrices, which implies that $(I_N\otimes  A- L_2\otimes BK)\otimes I_p$ is a {\it Hurwitz} matrix.
	
	Next, it is proved that $\|\hat \Xi_F(t)\|_2$ is convergent to zero.
	 {Since $r\le n-1$, it is obtained that $\D^n\delta_i(t)=\textbf{0}_p$}.
	Therefore, the solution to \eqref{eq:9} is $ {\hat \Xi}_F(t)=e^{((I_N\otimes  A- L_2\otimes BK)\otimes I_p)t}\hat\Xi_F(0)$.
	Since $(I_N\otimes  A- L_2\otimes BK)\otimes I_p$ is a {\it Hurwitz} matrix, it is obtained
	\begin{multline}
	\lim_{t\to\infty}\|\hat\Xi_F(t)\|_2\\
	\le\lim_{t\to\infty}\|e^{((I_N\otimes  A- L_2\otimes BK)\otimes I_p)t}\|_2\|\hat\Xi_F(0)\|_2=0.
	\end{multline}
\end{proof}

\section{Extensions to Followers with High-Order Integral Dynamics}\label{sec:highorder}
In Section \ref{sec:dis}, the followers are described by the first-order integral dynamics.
However, due to the diversity of control objects in practice, it is more interesting to study the follower described by the high-order integral dynamics. In this section, the dynamics of the $(i-M)$th follower $(i\in{\cal N})$ is described by
\begin{equation}\label{eq:7}
 {\D^{m}x_i(t)=u_i(t)},
\end{equation}
where  {$x_i(t)\in{\mathbb R}^p$} is the position vector of the $(i-M)$th follower; and  {$u_i(t)\in{\mathbb R}^p$} is the control input of the $(i-M)$th follower.

Motivated by the $PI^n$-type algorithm \eqref{eq:11}, we propose the following containment algorithm
\begin{equation}\label{eq:2}
	 {u_i(t)=\sum_{l=0}^{l_m-1}\kappa_l\D^{m-l-1}\sum_{j\in{\cal M\cup N}}\alpha_{ij}e_{ji}(t),}
\end{equation}
where $l_m=\max\{m,n+1\}$.

\
\begin{thm}\label{thm:con_2}
	Assume all leaders move along their polynomial trajectories described by \eqref{eq:16}; and all followers are described by \eqref{eq:7}.
	Let $P$ denote the positive definite solution to the following matrix inequality
	\begin{equation}
	 E^TP+PE+I_{l_m}-PFF^TP\le 0,\nonumber
	\end{equation}
	where
	\begin{equation}
	E=
	\begin{bmatrix}
	0 & 1 & \cdots & 0\\
	\vdots & \vdots & \ddots & \vdots\\
	0 & 0 & \cdots & 1\\
	0 & 0 & \cdots & 0
	\end{bmatrix}\in {\mathbb R}^{l_m\times l_m},\ \text{and} \
	F=
	\begin{bmatrix}
	0\\
	\vdots\\
	0\\
	1
	\end{bmatrix}\in{\mathbb R}^{l_m}.\nonumber
	\end{equation}
	The containment problem of MASs can be solved by \eqref{eq:2} with $K=(\kappa_{l_m-1},\cdots,\kappa_0)=\varepsilon F^TP$, where $\varepsilon$ is defined in Theorem \ref{thm:con_1}.
\end{thm}
	\color{black}
\begin{proof}
	Let  $\xi_i(t)=\big(x_i^T(t), {x_i^{(1)}}^T(t),\cdots,{x_i^{(l_m-1)}}^T(t)\big)^T$, $\Xi_L(t)=(\xi_1^T(t),\cdots,\xi_{M}^T(t)^T)^T$, $\Xi_F(t)=(\xi_{M+1}^T(t),\cdots,\xi_{M+N}^T(t))^T$ and $\hat \Xi_F(t)=\Xi_F(t)+(L_2^{-1}L_1\otimes I_{l_mp}) \Xi_L(t)$.
	 {
	Following the same procedure of the proof of Theorem \ref{thm:con_1}, it can be proved that there must exist two positive constants $M_1<\infty$ and $\beta_1$, such that $\|e^{(I_N\otimes E\otimes I_p- L_2\otimes FK\otimes I_p)t}\|_2\le M_1e^{-\beta_1t}$ and $\|\hat\Xi_F(t)\|_2\le M_1e^{-\beta_1t}\|\hat\Xi_F(0)\|_2\to 0\ (t\to\infty)$.}
	By Lemma \ref{lem:Laplacian}, the algorithm defined by \eqref{eq:2} solves the containment problem.
\end{proof}

\color{black}

\begin{remark}
It can be seen from \eqref{eq:2} that besides the proportional term  {$\sum_{j\in{\cal M\cup N}}\alpha_{ij}e_{ji}(t)$} and the integral terms {$\big\{\D^{-l}\sum_{j\in{\cal M\cup N}}\alpha_{ij}e_{ji}(t),\ l=1,2,\cdots,l_m-m\big\}$}, \eqref{eq:2} also includes the ``differential terms''  {$\big\{\D^l\sum_{j\in{\cal M\cup N}}\alpha_{ij}e_{ji}(t),\ l=1,2,\cdots,m-1\big\}$}.
Therefore, the proposed algorithm defined by \eqref{eq:2} is essentially a generalized ``$PID$'' algorithm (we can call it the ``$PI^{l_m-m}D^{m-1}$''-type algorithm).
\end{remark}

It is well known that the proportional term depends on the present information; the integral term represents the accumulation of past information; and the differential term is the future information, which might be more expensive to be measured.
Hence the differential terms  {$\big\{\D^l\sum_{j\in{\cal M\cup N}}\alpha_{ij}e_{ji}(t),\ l=1,2,\cdots,m-1\big\}$} in \eqref{eq:2} might be difficult to obtain.
Motivated by \cite{Li11IJRNC,Wen13IJC}, one way to handle this challenge is to design the state estimator for the $(i-M)$th follower agent to estimate its own state {$\{x_i(t), \D x_i(t),\cdots,\D^{m-1} x_i(t),\ i\in{\cal N}\}$}.
Let  {$z_i(t)=(z_{i}^T(t),\cdots,z_{im}^T(t)^T)^T\in{\mathbb R}^{mp}$} denote the estimated state of the $(i-M)$th follower agent.
Then the followers exchange their estimated states with their neighbor agents via the communication network $\cal G$ to obtain the differential terms.
Since the leaders are essentially the reference signals (the polynomial trajectories defined by \eqref{eq:16}), each leader should know its current position and any-order derivatives of the current position accurately. Therefore, there is no need to design the estimators for leaders. The $i$th leader ($i\in{\cal M}$) directly sends its state  {$z_i(t) = (x_i(t), \D x_i^T(t), \cdots,\D^{m-1} x_i^T(t))^T$} to the connected neighbors.

By the above discussion, the state estimator of the $(i-M)$th follower ($i\in{\cal N}$) is designed as
\begin{multline}\label{eq:48}
\D z_i(t)=(\bar E\otimes I_p)z_i(t)+(\bar F\otimes I_p) u_i(t)+(K_e\otimes I_p)\\
\times\sum_{j=1}^{M+N}\alpha_{ij}\Big((G\otimes I_p)(z_j(t)-z_i(t))- {e_{ji}(t)}\Big),
\end{multline}
where $K_e\in{\mathbb R}^m$, $\bar F=(0,\cdots,0,1)^T\in{\mathbb R}^{m}$, $G=(1,0,\cdots,0)\in{\mathbb R}^{1\times m}$ and
\begin{equation}
\bar E=
\begin{bmatrix}
0 & 1 & \cdots & 0\\
\vdots & \vdots & \ddots & \vdots\\
0 & 0 & \cdots & 1\\
0 & 0 & \cdots & 0
\end{bmatrix}\in {\mathbb R}^{m\times m}.\nonumber
\end{equation}

\begin{lemma}\label{lem:1}
	Let $K_e=\varepsilon PG^T$ where  {$\varepsilon$ is defined in Theorem \ref{thm:con_1} and} $P$ is the solution to the following matrix inequality
	\begin{equation}
	\bar EP+P\bar E^T+I_{l_m}-PG^TGP\le 0.\nonumber
	\end{equation}
	Then there exist two positive constants $M_2<\infty$ and $\beta_2$ such that $\|\D^{l_m-m}\hat Z(t)\|_2\le M_2e^{-\beta_2 t}$, where $\hat Z(t)=(\hat z_{M+1}^T(t),\cdots,\hat z_{M+N}^T(t))^T$, $\hat z_i(t)=z_i(t)-\zeta_i(t)$ and  {$\zeta_i(t)=(x_i^T(t), \D x_i^T(t),\cdots,\D^{m-1} x_i^T(t))^T$}.
\end{lemma}

\begin{proof}
	From \eqref{eq:7} and \eqref{eq:48}, it is obtained that
	\begin{equation}
	\D^{l_m-m+1}\hat Z(t)=(I_{N}\otimes E \otimes I_p-L_2\otimes K_eG\otimes I_p)\D^{(l_m-m)}\hat Z(t).\nonumber
	\end{equation}
	Following the same procedure of the proof of Theorem \ref{thm:con_1}, it can be easily proved that $(I_{N}\otimes E \otimes I_p-L_2\otimes K_eG\otimes I_p)$ is a {\it Hurwitz} matrix.
	Therefore, there must exist two positive constants  $M_2<\infty$ and $\beta_2$ such that $\|\D^{(l_m-m)}\hat Z(t)\|_2=\|e^{(I_{N}\otimes E \otimes I_p-L_2\otimes K_eG\otimes I_p)t}\|_2\|\hat Z^{(l_m-m)}(0)\|_2\le M_2e^{-\beta_2t}$.
\end{proof}

Replacing  {$\D^le_{ji}(t)$ with $(z_{jl}(t) - z_{il}(t))$, $l=1,\cdots,m-1$}, the containment algorithm \eqref{eq:2} is modified as
 {
\begin{multline}\label{eq:40}
u_i(t)=\sum_{j\in{\cal M\cup N}}\alpha_{ij}\left(\sum_{l=m-1}^{l_m-1}\kappa_l\D^{m-l-1}e_{ji}(t)\right.,\\
\left.+\sum_{l=0}^{m-2}\kappa_l(z_{j(m-l)}(t)-z_{i(m-l)}(t))\right).
\end{multline}
}

\

\begin{thm}\label{thm:con_3}
	Assume all leaders move along their polynomial trajectories described by \eqref{eq:16}; and all followers are described by high-order integral dynamics \eqref{eq:7}.
	The containment problem can be solved by \eqref{eq:40} with $K=(\kappa_{l_m-1},\cdots,\kappa_0)=\varepsilon F^TP$, where $\varepsilon$ is defined in Theorem \ref{thm:con_1} and $P$ is defined in Theorem \ref{thm:con_2}.
\end{thm}
\color{black}
\begin{proof}
	By applying the containment algorithm defined by \eqref{eq:40}, it can be obtained that
	\begin{equation}\label{eq:46}
	\D{\hat \Xi}_F(t)=R_1\hat\Xi_F(t)-R_2\D^{(l_m-m)}\hat Z(t),
	\end{equation}
	where $\hat \Xi_F(t)$ is defined in the proof of Theorem \ref{thm:con_2}; $R_1=I_N\otimes E\otimes I_p- L_2\otimes FK\otimes I_p$; $R_2= L_2\otimes FK_2\otimes I_p$ and $K_2=(0,\kappa_{m-2},\cdots,\kappa_0)$.
	
	The solution to \eqref{eq:46} is $\hat \Xi_F(t)=e^{R_1t}\hat\Xi_F(0)+\int_0^t e^{R_1(t-\tau)}R_2 \D^{(l_m-m)}\hat Z(\tau)d\tau$, which together with Theorem \ref{thm:con_2} and Lemma \ref{lem:1} leads to
	\begin{align*}
	\|\hat \Xi_F(t)\|_2\le&\|e^{R_1t}\|_2\|\hat\Xi_F(0)\|_2\\
	&+\int_0^t \|e^{R_1(t-\tau)}\|_2\|R_2\|_2\|\D^{(l_m-m)}\hat Z(\tau)\|_2d\tau\\
	\le& \Sigma_1(t)+\Sigma_2(t),
	\end{align*}
	where $\Sigma_1(t)=M_1e^{-\beta_1t}\|\hat \Xi_F[0]\|_2$ and $\Sigma_2(t)=M_1M_2\|R_2\|_2\int_0^te^{-\beta_1(t-\tau)}e^{-\beta_2\tau}d\tau$.
	It is easy to see that $\lim_{t\to\infty}\Sigma_1(t)=0$ and
	\begin{multline}
	\lim_{t\to\infty}\Sigma_2(t)=\\
	\begin{cases}
	M_1M_2\|R_2\|_2\lim\limits_{t\to\infty}\frac{e^{-\beta_1t}-e^{-\beta_2t}}{\beta_1-\beta_2}=0,& \text{if}\ \beta_2\neq \beta_1\\
	M_1M_2\|R_2\|_2\lim\limits_{t\to\infty}te^{-\beta_1t}=0,& \text{if}\ \beta_1=\beta_2
	\end{cases}.
	\end{multline}
	Therefore, $\lim_{t\to\infty}\|\hat \Xi_F(t)\|_2=0$. By Lemma \ref{lem:Laplacian}, the containment problem is solved.
\end{proof}

\color{black}
\begin{remark}
Since the leader sends its absolute position to the connected followers, one may wonder whether these followers can calculate their own absolute positions by their relative positions with the leader and the leader's absolute position. However, this idea does not work because the follower can receive the position information not only from the leader (accurate position) but also other followers (estimated positions, not accurate). Due to the nature of distributed control of MASs, the follower cannot distinct the leader from other neighbor agents. Hence the follower can only randomly pick up one agent in its neighborhood to calculate its absolute position. If the selected agent is another follower, the calculated absolute position is obviously inaccurate.
\end{remark}

\
Next, a simulation example is provided to demonstrate the effectiveness of the proposed algorithm.

\noindent\textbf{Simulation Example:}
\color{black}
Consider a MAS composed of eight agents, whose communication topology is shown in Fig. \ref{fig:sim2_top}.
It is easy to see that agents $1$ to $4$ are leaders and agents $5$ to $8$ are followers.
 {The $i$th row and the $j$th column entry of the adjacency matrix satisfies that $\alpha_{ij}=1$ if there is a directed edge from agent $j$ to agent $i$, otherwise $\alpha_{ij}=0$.}

\begin{figure}
	\centering
	\includegraphics[width=0.7\hsize]{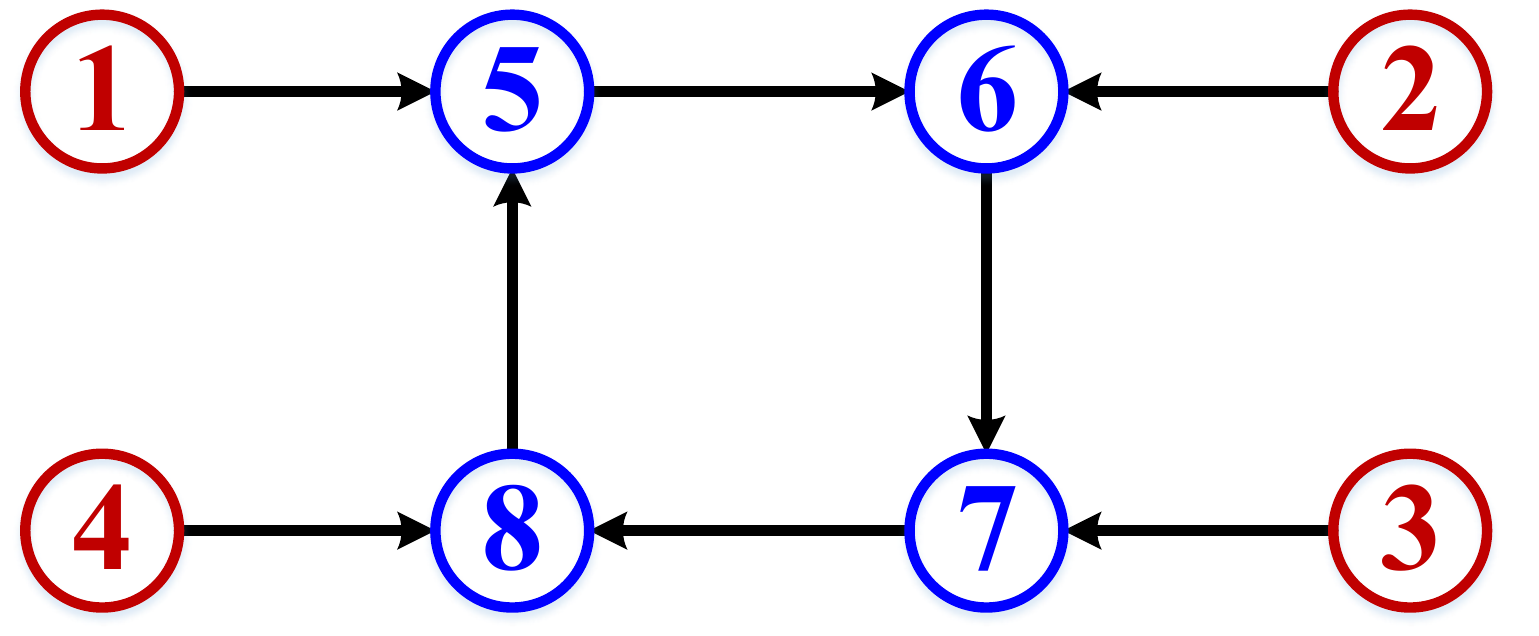}
	\caption{The interaction topology of the MAS in the Simulation Example of Section \ref{sec:highorder}.}
	\label{fig:sim2_top}
\end{figure}

 {
The $i$th leader's move along the trajectory defined by the following polynomial}
\begin{equation}\label{eq:51}
x_i(t)={\bf a}^i_0+{\bf a}^i_1t+{\bf a}^i_2t^2+{\bf a}^i_{3}t^{3},
\end{equation}
where $x_i(t)\in{\mathbb R}^2$ and the coefficients ${\bf a}_j^i\in{\mathbb R}^2$ are  {given} in TABLE \ref{Tab:sim1_coef}.
\begin{table}
	\centering\caption{Values of Coefficients $a^i_j$ in \eqref{eq:51}.}\label{Tab:sim1_coef}
	\begin{tabular}{c|c|c|c|c}
		\hline\hline
		agent & 1 & 2 & 3 & 4 \\\hline
		$a^i_0$ & $\begin{pmatrix}0\\0 \end{pmatrix}$ & $\begin{pmatrix}2\\5 \end{pmatrix}$ & $\begin{pmatrix}-5\\10 \end{pmatrix}$ &  $\begin{pmatrix}-10\\0 \end{pmatrix}$\\\hline
		$a^i_1$ & $\begin{pmatrix}0.23\\3.43 \end{pmatrix}$  &$\begin{pmatrix}0.3\\3.43 \end{pmatrix}$ & $\begin{pmatrix}0.2\\3.43 \end{pmatrix}$ & $\begin{pmatrix}0.2\\3.43 \end{pmatrix}$ \\\hline
		$a^i_2$ & $\begin{pmatrix}0.0095\\-0.75\end{pmatrix}$ & $\begin{pmatrix}0.0095\\-0.075 \end{pmatrix}$ & $\begin{pmatrix}0.01\\-0.075 \end{pmatrix}$ & $\begin{pmatrix}0.01\\-0.075 \end{pmatrix}$ \\\hline
		$a^i_3$& $\begin{pmatrix}0\\0.0005 \end{pmatrix}$ & $\begin{pmatrix}0\\0.0005 \end{pmatrix}$ & $\begin{pmatrix}0\\0.0005 \end{pmatrix}$ & $\begin{pmatrix}0\\0.0005 \end{pmatrix}$\\
		\hline\hline
	\end{tabular}
\end{table}
The $(i-4)$th follower ($i=5,6,7,8$) has the third-order integral dynamics, \emph{i.e.},
$x_{i}^{(3)}(t)=u_i(t)$.

{The control algorithm \eqref{eq:40} is applied to solve this containment problem.
By Theorem \ref{thm:con_2}, the parameters in algorithm \eqref{eq:40} are set to be $K=(\kappa_3,\kappa_2,\kappa_1,\kappa_0)=(2,\ 6.1554,\ 8.4721,\  6.1554)$.}
The simulation result is given in  Fig. \ref{fig:sim2_res}.
As shown in Fig. \ref{fig:sim2_res}, all followers are convergent into the convex hull spanned by the leaders and move along with them.
Therefore, the proposed algorithm is able to effectively solve the containment problem with dynamic leaders.

\begin{figure}
	\centering
	\includegraphics[width=0.9\hsize]{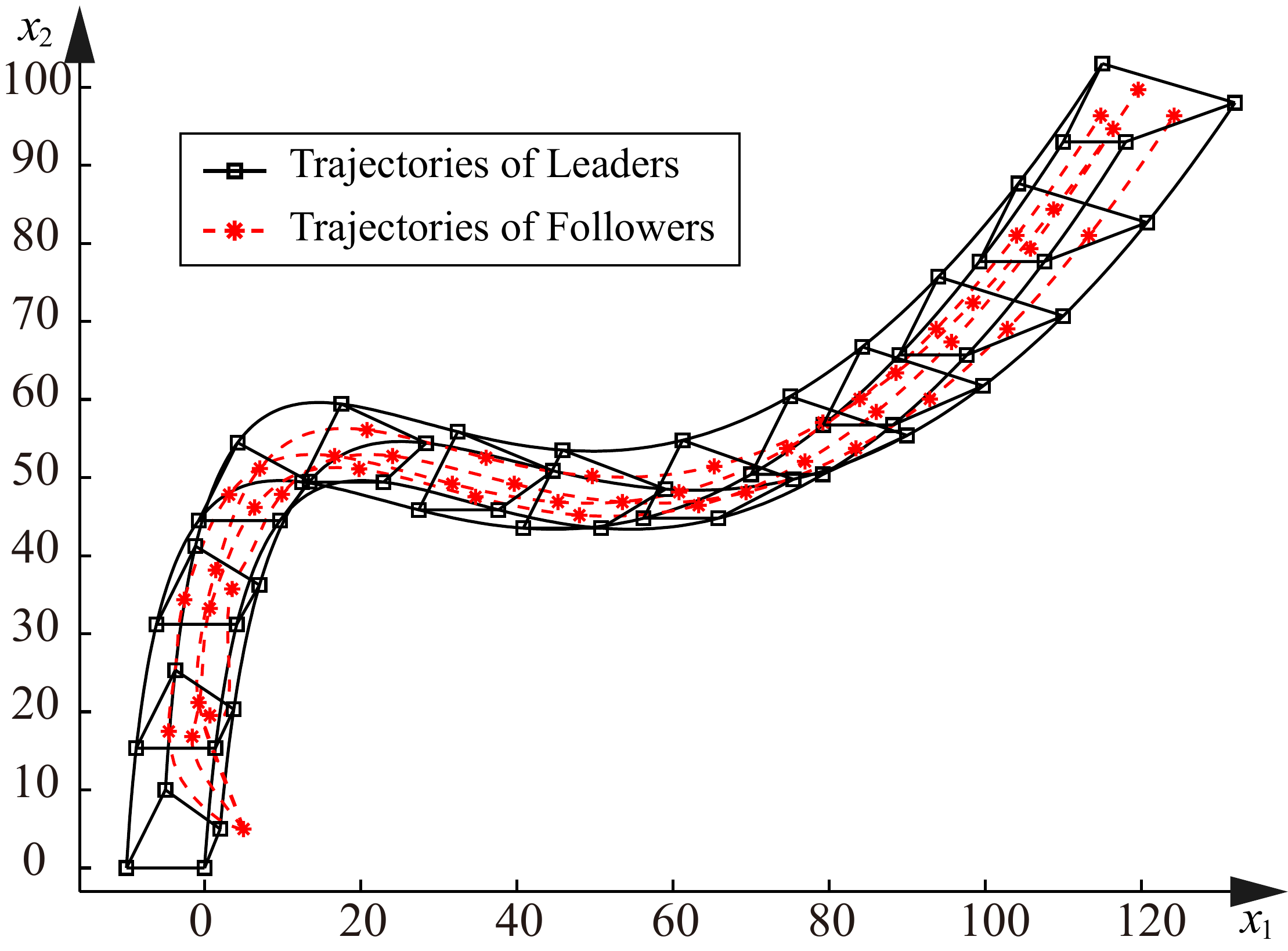}\\
	\caption{Moving profiles of all agents in the Simulation Example of Section \ref{sec:highorder}.}\label{fig:sim2_res}
\end{figure}

\color{black}
\section{ {Extensions to MASs in Discrete-Time Domain}}\label{sec:ext}
In this subsection, the containment problem is studied in the discrete-time domain.
The $i$th leader is assumed to move along the following polynomial trajectory
\begin{equation}\label{eq:30}
x_i[k]={\bf a}_0^i+{\bf a}_1^ik+\cdots+{\bf a}^i_{n}k^{n},
\end{equation}
where $x_i[k]\in{\mathbb R}^p$ and ${\bf a}^i_j\in{\mathbb R}^p$.

{In this section, the symbol $\D$ denotes the difference operator (namely, $\D x[k]=x[k+1]-x[k]$). The inversion of the difference operator $\D$ is defined as $\D^{-1}x[k]=\sum_{i=0}^{k-1}x[k]$. It is easy to see that $\D(\D^{-1}x[k])=x[k]$.
}

%The rest of this section is divided into two subsections according to the followers' dynamics.

\subsection{Followers with Single-Integrator Dynamics}\label{sec:con_1}
In this subsection, the $(i-M)$th follower $(i\in{\cal N})$ is described by the following single-integrator dynamics
\begin{equation}\label{eq:23}
 {\D x_i[k]=u_i[k]+\delta_i[k],}
\end{equation}
where $x_i[k]\in{\mathbb R}^p$ denotes the position of the $(i-M)$th follower; $u_i[k]\in{\mathbb R}^p$ is the control input,  {$\delta_i[k]={\bf b}_0^i+{\bf b}_1^ik+\cdots+{\bf b}_r^ik^r\in{\mathbb R}^p$ is the polynomial disturbance, and ${\bf b}_j^i\in{\mathbb R}^p\;(j=0,\cdots,r) $}.
Motivated by \eqref{eq:11}, the following discrete-time $PI^n$-type algorithm is proposed
 {
\begin{multline}\label{eq:24}
u_i[k]=\frac{1}{1+d_i}\sum_{l=0}^{n}\kappa_l\D^{-l}\bigg(\sum_{j\in{\cal M}}\alpha_{ij}e_{ji}[k]\\
+\sum_{j\in{\cal N}}\alpha_{ij}e_{ji}[k]\bigg),
\end{multline}
where $e_{ji}[k]=x_j[k]-x_i[k]$.}

Let $\xi_i[k]=(x_i^T[k],\D x_i^T[k],\cdots,\D^nx_i^T[k])^T$, $\Xi_L[k]=(\xi_1^T[k],\cdots,\xi_{M}^T[k])^T$ and $\Xi_F[k]=(\xi_{M+1}^T[k],$ $\cdots,\xi_{M+N}^T[k])^T$.
Substituting \eqref{eq:24} into \eqref{eq:23} obtains the following closed-loop dynamics
\
\begin{equation}\label{eq:36}
\begin{cases}
\D\Xi_L[k]=&(I_M\otimes A\otimes I_p)\Xi_L[k], \\
\D\Xi_F[k]=&\big(I_N\otimes A\otimes I_p\\
&-(I_N+D)^{-1}L_2\otimes BK\otimes I_p\big)\Xi_F[k]\\
&-\big((I_N+D)^{-1}L_1\otimes BK\otimes I_p\big)\Xi_L[k]\\
&+(I_N\otimes B \otimes I_p)\Delta[k],
\end{cases}
\end{equation}
where $D=\diag(d_{M+1},\cdots,d_{M+N})$, $\Delta[k]=\diag(\D^n\delta_{M+1}^T[k],\cdots,\D^n\delta_{M+N}^T[k])$,  \color{black} $K=(\kappa_{n-1},\cdots,\kappa_0)$; $A$ and $B$ are defined in \eqref{eq:3}.
It follows from \eqref{eq:36} that
\begin{multline}\label{eq:25}
 {{\hat \Xi}_F[k+1]=(I_N\otimes \hat A\otimes I_p-\hat L_2\otimes BK\otimes I_p)\hat\Xi_F[k]}\\
 {+(I_N\otimes B \otimes I_p)\Delta[k],}
\end{multline}
where $\hat \Xi_F[k]=\Xi_F[k]+(L_2^{-1}L_1\otimes I_{(n+1)p}) \Xi_L[k]$, $\hat A=A+I_{n+1}$ and $\hat L_2=(I_N+D)^{-1}L_2$.
By Lemma \ref{lem:Laplacian}, if every element in $\hat\Xi_F[k]$ is convergent to zero, then the containment problem is solved.
\
\begin{thm}\label{thm:dis_1}
	Assume all leaders move along their polynomial trajectories described by \eqref{eq:30}; and all followers are described by the single-integrator dynamics \eqref{eq:23}.
	Let $P$ denote the positive definite solution to the following matrix inequality
	\begin{equation}\label{eq:5}
		P>\hat A^TP\hat A-(1-\varepsilon^2)\hat A^TPB(B^TPB)^{-1}B^TP\hat A,
	\end{equation}
	where $\varepsilon\in(\max_{i\in\{1,2,\cdots,N\}}|1-\hat \lambda_i|,1)$, and $\{\hat \lambda_i; i=1,\cdots,N\}$ are the eigenvalues of $\hat L_2$.
	If the order of the polynomial disturbance $\delta_i[k]$ in \eqref{eq:23} is not greater than $n-1$ (namely, $r\le n-1$), then the containment problem in the discrete-time domain can be solved by the $PI^n$-type algorithm defined by \eqref{eq:24} with $K=(B^TPB)^{-1}B^TP\hat A$.
\end{thm}
\color{black}
\begin{proof}
	 Since $r\le n-1$, it can be obtained that $\D^n\delta_i[k]=\textbf{0}_p$. This together with \eqref{eq:25} leads to that
	\begin{equation*}
	\hat\Xi_F[k]=(I_N\otimes \hat A\otimes I_p-\hat L_2\otimes BK\otimes I_p)^k\hat\Xi_F[0].
	\end{equation*}

	By the Gerschgorin circle theorem and Lemma \ref{lem:Laplacian}, all eigenvalues of $\hat L_2$ are inside the open circle $U(1,1)\triangleq\{a+bj|a,b\in{\mathbb R}\ \text{and}\ (a-1)^2+b^2<1\}$.
	Therefore, the set $(\max_{i\in\{1,2,\cdots,N\}}|1-\lambda_i|,1)$ is not empty.
	By Lemma 5 in \cite{Kristian13Automatica}, we know that the matrix inequality \eqref{eq:5} has a positive definite solution $P$ as long as $|\varepsilon|<1/\prod_i|\lambda_i^u(\hat A)|$, where $\{\lambda_i^u\}$ are the unstable eigenvalues of $\hat A$.
	Since all unstable eigenvalues of $\hat A$ are $1$, \eqref{eq:5} has a positive definite solution $P$.
	
	Next, it is proved all eigenvalues of $\hat A-\hat \lambda_iBK$ are inside the unit circle.
	From \eqref{eq:5}, it can be calculated that
	\begin{IEEEeqnarray}{lll}
		&(\hat A-\hat \lambda_i BK)^HP(\hat A-\hat \lambda_i BK)-P\IEEEnonumber\\
		=&(\hat A^TP\hat A-P)-(\hat \lambda_i+\bar {\hat\lambda}_i-\hat \lambda_i\bar{\hat \lambda}_i )\hat A^TPB(B^TPB)^{-1}B^TP\hat A\IEEEnonumber\\
		<&(1-\delta^2)\hat A^TPB(B^TPB)^{-1}B^TP\hat A\IEEEnonumber\\
		&-(1-(1-\hat \lambda_i)(1-\bar {\hat \lambda}_i))\hat A^TPB(B^TPB)^{-1}B^TP\hat A\IEEEnonumber\\
		=&(|1-\hat \lambda_i|^2-\varepsilon^2)\hat A^TPB(B^TPB)^{-1}B^TP\hat A\le0.\nonumber
	\end{IEEEeqnarray}
	By Lyapunov stability theory, all eigenvalues of $\hat A-\hat \lambda_iBK$ are inside the unit circle.

	By Schur decomposition, there must exist a transformation matrix $T$ such that
	\begin{equation}
	\Lambda\triangleq T\hat L_2T^{-1}=
	\begin{bmatrix}
	\lambda_1& * &\cdots&* \\
	0 &\lambda_2 & \cdots&*\\
	\vdots &\vdots&\ddots&\vdots\\
	0 & 0&\cdots&\lambda_N
	\end{bmatrix}.\nonumber
	\end{equation}
	Hence,
		\begin{multline}
		I_N\otimes \hat A\otimes I_p-\hat L_2\otimes BK\otimes I_p\\
		=((T^{-1}\otimes I_{n})(I_N\otimes \hat A-\Lambda\otimes BK)(T\otimes I_{n}))\otimes I_p.
		\end{multline}
	The diagonal elements of $I_N\otimes \hat A-\Lambda\otimes BK$ are $\hat A-\hat \lambda_iBK$ ($i=1,\cdots,N$).
	Therefore, all eigenvalues of $I_N\otimes \hat A-\Lambda\otimes BK$ are inside the unit circle.
	Hence, there must exist two positive constants $M_3<\infty$ and $\beta_3\in(0,1)$ such that
	$\|(I_N\otimes \hat A\otimes I_p-\hat L_2\otimes BK\otimes I_p)^k\|_2\le M_3\beta_3^k$ and
	$\|\hat\Xi_F[k]\|_2\le M_3\beta_3^k\|\Xi_F(0)\|_2\to0\ (k\to\infty)$.
\end{proof}

In \eqref{eq:24}, the coefficient $1/(1+d_i)$ is used to normalize the Laplacian matrix.
If we replace $1/(1+d_i)$ with a uniform constant $\mu$, then the algorithm \eqref{eq:24} is modified as
\begin{equation}\label{eq:8}
	u_i[k]=\mu\sum_{l=0}^{n}\kappa_l\D^{-l}\bigg(\sum_{j\in{\cal M}}\alpha_{ij}e_{ji}[k]+\sum_{j\in{\cal N}}\alpha_{ij}e_{ji}[k]\bigg).
\end{equation}

Let $\{\lambda_1,\cdots,\lambda_N\}$ denote the eigenvalues of $L_2$.
Then by Lemma \ref{lem:Laplacian}, we know $\Re(\lambda_i)>0$.
Hence, there must exist a positive constant $\mu_{\min}\in(0,1)$ such that $\max_{i}|1-\mu_{\min}\lambda_i|<1$.
Following the same procedure of Theorem \ref{thm:dis_1}, the following corollary can be easily obtained.
\begin{coro}
	Let $P$ denote the positive definite solution to the following matrix inequality
	\begin{equation}
	P>\hat A^TP\hat A-(1-\gamma^2)\hat A^TPB(B^TPB)^{-1}B^TP\hat A,
	\end{equation}	
	where $\gamma\in(\varepsilon_{\max},1)$ and $\varepsilon_{\max}=\max_{i}\{|1-\mu_{\min}\lambda_i|\}<1$.
	The containment problem of MASs in the discrete time domain can be solved by the algorithm \eqref{eq:8} with $K=(B^TPB)^{-1}B^TP\hat A$ and $\mu=\mu_{\min}$.
\end{coro}

 {
	\begin{remark}
		Compared to the previous results \cite{Cao12Automatica,Cao11TCST,Li12TAC,Zhang14SCL,Mei12Automatica,Meng10Automatica}, one distinguished feature of this paper is that the proposed algorithms do not employ the ``sign'' function.
		One limitation of the ``sign'' function is that it can cause the harmful ``chattering'' phenomenon. Therefore, the proposed algorithm can avoid the high-frequency control switches occurred in the ``chattering'' phenomenon.
		Moreover, the ``sign'' function can hardly be used in the discrete-time domain. The current literature rarely discusses the containment problem with dynamic leaders in the discrete-time domain, and this paper gives a preliminary attempt.
	\end{remark}
}
	
\subsection{Containment Control of MASs with Measurement Noises}\label{sec:noise}
In the above sections, it is assumed that all followers can accurately measure the relative states between themselves and their neighbors.
However, the measurement noise is unavoidable in practice.
In this subsection, it is assumed that the relative state $e_{ji}[k]$ is corrupted by the measurement noise $\rho_{ji}\eta_{ji}[k]$, where $\rho_{ji}=\diag(\rho_{ji1},\cdots,\rho_{jip})$, $\rho_{ijs}<\infty\ (s=1,\cdots,p)$ denotes the noise intensity; and $\eta_{ij}[k]\in{\mathbb R}^p$ is the standard white noise vector.
Moreover, it is assumed that $\{\eta_{ji}[k]|j\in{\cal M\cup N};i\in{\cal N}\}$ are mutually independent.

Denote $\bar e_{ji}[k]=e_{ji}[k]+\rho_{ji}\eta_{ji}[k]$.
The containment $PI^n$-type algorithm \eqref{eq:24} is modified as
\begin{multline}\label{eq:12}
u_i[k]=\frac{1}{1+d_i}\sum_{l=0}^{n}\kappa_l\D^{-l}\bigg(\sum_{j\in{\cal M}}\alpha_{ij}\bar e_{ji}[k]\\
+\sum_{j\in{\cal N}}\alpha_{ij}\bar e_{ji}[k]\bigg).
\end{multline}
The definition of the containment problem should also be modified to take the noise effect into account.
\begin{definition}\label{def:2}
	The containment problem of MASs with measurement noises is solved in the stochastic sense if
$\lim_{t\to\infty}\dis(\E(x_i[k]),co_L[k])=0$
	and $\E(\dis(x_i[k],co_L[k]))^2<\infty,\ i\in{\cal N}$.
\end{definition}

\begin{thm}
	Assume all leaders move along their polynomial trajectories described by \eqref{eq:30}; and all followers are described by the single-integrator dynamics \eqref{eq:23}.
	If the order of the polynomial disturbance $\delta_i[k]$ in \eqref{eq:23} is not greater than $n-1$ (namely, $r\le n-1$), then the containment algorithm \eqref{eq:12} with $K=(B^TPB)^{-1}B^TP\hat A$ ($P$ is defined in Theorem \ref{thm:dis_1}) can solve the containment problem of MASs with measurement noises in the stochastic sense.
\end{thm}
\begin{proof}
	Since $r\le n-1$, the closed-loop dynamics \eqref{eq:25} can be rewritten as
	\begin{multline}\label{eq:10}
	{\hat \Xi}_F[k+1]=(I_N\otimes \hat A\otimes I_p-\hat L_2\otimes BK\otimes I_p)\hat\Xi_F[k]\\+((I_N+D)^{-1}\Gamma\otimes BK\otimes I_p) V[k],
	\end{multline}
	where $\Gamma=\diag(\Gamma_{M+1},\cdots,\Gamma_{M+N})$, $V[k]=$ $(V_{M+1}^T[k],\cdots,V_{M+N}^T[k])^T$,  $\Gamma_i=(\alpha_{i1}\rho_{1i},\cdots,$ $\alpha_{i(M+N)}\rho_{(M+N)i})$, $V_i[k]=$ $(\nu_{1i}^T[k],\cdots,$ $\nu_{(M+N)i}^T[k])^T$, and $\nu_{ji}[k]=(\eta_{ji}[k],\D \eta_{ji}[k],\cdots,\D^n\eta_{ji}[k])^T$.
	
	Let $R_3=I_N\otimes \hat A\otimes I_p-\hat L_2\otimes BK\otimes I_p$ and $R_4=(I_N+D)^{-1}\Gamma\otimes BK\otimes I_p$.
	By \eqref{eq:10}, it is obtained that
	\begin{equation}
	{\hat \Xi}_F[k]=R_3^k{\hat \Xi}_F[0]+\sum_{i=0}^{k-1}R_3^{k-i-1}R_4V[i].
	\end{equation}
	From the proof of Theorem \ref{thm:dis_1}, we know that there exist two finite positive constants $M_3$ and $\beta_3\in(0,1)$ such that $\|R_3^k\|_F\le M_3\beta_3^k$.
	Hence $\lim_{k\to\infty}R_3^k{\hat \Xi}_F[0]=\textbf{0}_{p(n+1)N}$.
	By Lemma \ref{lem:r3}, we know that $\E(\sum_{i=0}^{k-1}R_3^{k-i-1}R_4V[i])=\sum_{i=0}^{k-1}R_3^{k-i-1}R_4\E(V[i])=\textbf{0}_{p(n+1)N}$.
	Therefore,
	\begin{equation}
	\lim_{k\to\infty}\E(\hat \Xi_F[k])=\lim_{k\to\infty}R_3^k{\hat \Xi}_F[0]=\textbf{0}_{p(n+1)N},
	\end{equation}
	which together with Lemma \ref{lem:Laplacian} leads to that
	$\lim_{t\to\infty}\dis(\E(x_i[k]),co_L[k])=0,\ i\in{\cal N}$.

	Let $Y_1[k]=R_3^k{\hat \Xi}_F[0]$ and $Y_2[k]=\sum_{i=0}^{k-1}R_3^{k-i-1}R_4V[i]$. Since $\lim_{k\to\infty}\|Y_1[k]\|_2=0$, there must exist a finite positive constant $\hat M_3$ such that $\|Y_1[k]\|_2<\hat M_3$.
	By Lemma \ref{lem:r2}, we know that $\forall m\in\{0,1,2\cdots,n\}$, there exists a finite positive constant $\bar M_3$ such that
	$\|\E(V[k]V^T[k+m])\|_F<\bar M_3$.
	
	And for $\forall m\ge n+1$,
	$\E(V[k]V^T[k+m])=\textbf{0}_{N(M+N)(n+1)p\times N(M+N)(n+1)}$ holds. Therefore,
	\begin{align*}
	&\|\E(Y_2[k]Y_2^T[k])\|_F\\
	&=\left\|\E \sum_{i=0}^{k-1}\sum_{j=0}^{k-1}R_3^{k-i-1}R_4V[i]V^T[j]R_4^T{R_3^{k-j-1}}^T\right\|_F\\
	&=\left\|\E\sum_{i=0}^{k-1}\sum_{j=\max\{0,i-n\}}^{\min\{k-1,i+n\}}\!\!\! R_3^{k-i-1}R_4V[i]V^T[j]R_4^T{R_3^{k-j-1}}^T\right\|_F\\
	&\le \sum_{i=0}^{k-1}\sum_{j=i-n}^{i+n}\!\!\!\! \|R_3^{k-i-1}\|_F\|R_4\|_F^2\|\E(V[i]V^T[j])\|_F\|{R_3^{k-j-1}}\|_F\\
	& \le M_3^2\bar M_3\|R_4\|_F^2\sum_{i=0}^{k-1}\sum_{j=i-n}^{i+n}\beta_3^{2k-i-j-2}\\
	&=M_3^2\bar M_3\|R_4\|_F^2 \frac{(1-\beta_3^{2k})(1-\beta_3^{2n+1})}{\beta_3^{n}(1-\beta_3)(1-\beta_3^2)}<\infty,
	\end{align*}
	which leads to that  $\E\|Y_2[k]\|_2^2\le \|\E(Y_2[k]Y_2^T[k])\|_F^2<\infty$.
	Therefore,
	\begin{equation}
	\E\|\hat \Xi_F[k]\|_2^2\le 2 \E\|Y_1[k]\|_2^2+2\E\|Y_1[k]\|_2^2<\infty,\nonumber
	\end{equation}
	which together with Lemma \ref{lem:Laplacian} leads to that $\E(\dis(x_i[k],co_L[k]))^2<\infty,\ i\in{\cal N}$.
\end{proof}

\color{black}
\subsection{Followers with High-Order Integral Dynamics}\label{sec:con_h}

In this subsection, the dynamics of the $(i-M)$th follower ($i\in{\cal N}$) is described by the following high-order difference equation
\begin{equation}\label{eq:32}
\D^{m}x_i[k]=u_i[k],
\end{equation}
where $x_i[k]\in{\mathbb R}^p$ is the position vector of the $(i-M)$th follower; and $u_i[k]\in{\mathbb R}^p$ is the control input.

Motivated by the algorithm \eqref{eq:2}, the following $PI^{l_m-m}D^{m-1}$-type containment algorithm is proposed
\begin{multline}\label{eq:49}
u_i[k]=\frac{1}{1+d_i}\sum_{l=0}^{l_m-1}\kappa_l\D^{m-l-1}\bigg(\sum_{j\in{\cal M}}\alpha_{ij}e_{ji}[k]\\
+\sum_{j\in{\cal N}}\alpha_{ij}e_{ji}[k]\bigg),
\end{multline}
where $l_m=\max\{n+1,m\}$.

\begin{thm}\label{thm:dis_2}
	Assume all leaders move along their polynomial trajectories described by \eqref{eq:30}; and all followers are described by the high-order integral dynamics \eqref{eq:32}.
	Let $P$ denote the positive definite solution to the following matrix inequality
	\begin{equation}
	P>\hat E^TP\hat E-(1-\varepsilon^2)\hat E^TPF(F^TPF)^{-1}F^TP\hat E, \nonumber
	\end{equation}
	where $\varepsilon\in(\max_{i\in\{1,2,\cdots,N\}}|1-\hat \lambda_i|,1)$, $\hat E=I_{l_m}+E$, $E$ and $F$ are defined in Theorem \ref{thm:con_2}.
	The containment problem can be solved by \eqref{eq:49} with $K=(\kappa_{l_m-1},\cdots,\kappa_0)=(F^TPF)^{-1}F^TP\hat E$.
\end{thm}
\color{black}
\begin{proof}
	Let  $\xi_i[k]=(x_i^T[k],\D x_i^T[k],\cdots,$ $\D^{l_m-1} x_i^T[k])^T$, $\Xi_L[k]=(\xi_1^T[k],\cdots,\xi_{M}^T[k])^T$, $\Xi_F[k]=(\xi_{M+1}^T[k],\cdots,\xi_{M+N}^T[k])^T$ and $\hat \Xi[k]=\Xi_F[k]+(L_2^{-1} L_1\otimes I_{l_mp})\Xi_L[k]$.
	 {
	Following the same procedure of the proof of Theorem \ref{thm:dis_1}, it can proved that there exist two positive constants $M_4<\infty$ and $\beta_4\in(0,1)$ such that $\|(I_N\otimes \hat E-\hat L_2\otimes FK)^k\|_2\le M_4\beta_4^k$ and $\lim_{k\to\infty}\|\hat\Xi[k]\|_2\le \lim_{k\to\infty} M_4\beta_4^k\|\hat \Xi[0]\|_2=0$.}
\end{proof}

If  {$\{\D e_{ji}[k],\cdots,\D^{m-1}e_{ji}[k]\}$} in \eqref{eq:49} are not available for the algorithm design, motivated by \eqref{eq:48}, the following estimator is designed to estimate the $(i-M)$th follower's position and the position's differences up to the $(m-1)$th-order, $i\in{\cal N}$.
\begin{multline}\label{eq:47}
\D z_i[k]=(\bar E\otimes I_p)z_i[k]+(\bar F\otimes I_p)u_i[k]+\frac{(K_e\otimes I_p)}{1+d_i}\\
\times\sum_{j=1}^{M+N}\alpha_{ij}\Big((G\otimes I_p)(z_j[k]-z_i[k])
-e_{ji}[k]\Big),
\end{multline}
where $z_i[k]=(z_{i1}^T[k],\cdots,z_{im}^T[k]^T)^T$;  $K_e$, $\bar E$, $\bar F$ and $G$ are defined in \eqref{eq:48}.

\begin{lemma}\label{lem:2}
	Let $K_e=\tilde EPG^T(GPG^T)^{-1}$ where $P$ is the positive definite solution to the following matrix inequality
	\begin{equation*}
	P>\tilde EP\tilde E^T-(1-\varepsilon^2)\tilde EPG^T(GPG^T)^{-1}GP\tilde E^T,
	\end{equation*}
	where $\varepsilon\in(\max_{i\in\{1,2,\cdots,N\}}|1-\hat \lambda_i|,1)$ and $\tilde E=I_m+\bar E$.
	Then there exist two positive constants $M_5<\infty$ and $\beta_5\in(0,1)$ such that $\|\D^{l_m-m}\hat Z[k]\|_2\le M_5\beta_5^k$, where $\hat Z[k]=(\hat z_{M+1}^T[k],\cdots,\hat z_{M+N}^T[k])^T$, $\hat z_i[k]=z_i[k]-\zeta_i[k]$ and $\zeta_i[k]=(x_i^T[k],\D x_i^T[k],\cdots,\D^{m-1}x_i^T[k])^T$.
\end{lemma}
\begin{proof}
	From \eqref{eq:32} and \eqref{eq:47}, it can be obtained that
	\begin{equation}
	\hat Z[k+1]=((I_{N}\otimes \tilde E -L_2\otimes K_eG)\otimes I_p)\hat Z[k],\nonumber
	\end{equation}
	which implies that
	\begin{equation}
	\D^{l_m-m}\hat Z[k+1]=((I_{N}\otimes \tilde E -L_2\otimes K_eG)\otimes I_p)\D^{l_m-m}\hat Z[k].\nonumber
	\end{equation}
	Following the same procedure of the proof of Theorem \ref{thm:dis_1}, it can be easily proved that all eigenvalues of $I_{N}\otimes \tilde E \otimes I_p-L_2\otimes K_eG\otimes I_p$ are inside the unit circle. Hence, there exist two positive constants $M_5<\infty$ and $\beta_5\in(0,1)$ such that $\|\D^{l_m-m}\hat Z[k]\|_2\le \|(I_{N}\otimes \tilde E -L_2\otimes K_eG)^k\|_2\|\D^{l_m-m} \hat Z[0]\|_2\\
	\le M_5\beta_5^k$.
\end{proof}

Replacing $\D^l e_{ji}[k]$ with $z_{jl}[k] - z_{il}[k]$ $(l=1,\cdots,m-1)$ in \eqref{eq:49} obtains the following modified containment algorithm:

{\begin{multline}\label{eq:6}
	u_i[k]=\frac{1}{1+d_i}\sum_{j\in{\cal M\cup N}}\alpha_{ij}\left(\sum_{l=m-1}^{l_m-1}\kappa_l\D^{m-l-1}e_{ji}[k]\right.,\\
	\left.+\sum_{l=0}^{m-2}\kappa_{m-l-1}(z_{j(m-l)}[k]-z_{i(m-l)}[k])\right).
	\end{multline}
	}

\
\begin{thm}\label{thm:dis_3}
	Assume all leaders move along their polynomial trajectories described by \eqref{eq:30}; and all followers are described by the high-order-integral dynamics \eqref{eq:32}.
	The containment problem of MASs can be solved by \eqref{eq:6} with $K=(\kappa_{l_m-1},\cdots,\kappa_0)=(F^TPF)^{-1}F^TP\hat E$, where $P$ is the solution to the following matrix inequality
	\begin{equation}
	P>\hat E^TP\hat E-(1-\varepsilon^2)\hat E^TPF(F^TPF)^{-1}F^TP\hat E, \nonumber
	\end{equation}
	where $\varepsilon\in(\max_{i\in\{1,2,\cdots,N\}}|1-\hat \lambda_i|,1)$.
\end{thm}

\color{black}

\begin{proof}
	By applying the containment algorithm defined by \eqref{eq:40}, the closed-loop dynamics of the MAS can be rewritten in the following compact form
	\begin{equation}\label{eq:41}
	\begin{cases}
	\D\Xi_L[k]=&\!\!\!\!(I_M\otimes E\otimes I_p)\Xi_L[k], \\
	\D\Xi_F[k]=&\!\!\!\!(I_N\otimes E\otimes I_p-\hat L_2\otimes FK\otimes I_p)\Xi_F[k]\\
	&-((I_N+D)^{-1}L_1\otimes FK\otimes I_p)\Xi_L[k]\\
	&-(\hat L_2\otimes FK_2\otimes I_p)\D^{l_m-m}\hat Z[k],
	\end{cases}
	\end{equation}
	where $K_2=(0,\kappa_{m-2},\cdots,\kappa_0)$; $\Xi_L[k]$ and $\Xi_F[k]$ are defined in Theorem \ref{thm:dis_2}.
	Let $\hat \Xi_F[k]=\Xi_F[k]+(L_2^{-1}L_1\otimes I_{l_mp}) \Xi_L[k]$. Then it is derived from \eqref{eq:41} that
	\begin{equation}\label{eq:42}
	{\hat \Xi}_F[k+1]=R_5\hat\Xi_F[k]-R_6\D^{l_m-m}\hat Z[k],
	\end{equation}
	where $R_5=I_N\otimes \hat E\otimes I_p-\hat L_2\otimes FK\otimes I_p$ and $R_6=\hat L_2\otimes FK_2\otimes I_p$.
	
	%	By the proof of Theorem \ref{thm:dis_1}, it is easy to see that there exist two positive constants $M_3<\infty$ and $\beta_3\in(0,1)$ such that $\|R_1^k\|_2\le M_3\beta_3^k$.
	It follows from \eqref{eq:42} that $\hat \Xi_F[k]=R_5^k\Xi_F[0]+\sum_{l=0}^{k-1}R_5^{k-l-1}R_6\D^{l_m-m}\hat Z[l]$.
	This together with the proof of Theorem \ref{thm:dis_2} and Lemma \ref{lem:2} implies that
	\begin{multline}
	\|\hat \Xi_F[k]\|_2\le
	\sum_{l=0}^{k-1}\|R_5^{k-l-1}\|_2\|R_6\|_2\|\D^{l_m-m}\hat Z[l]\|_2\\ + \|R_5^k\|_2\|\Xi_F[0]\|_2
	\le\Sigma_1[k]+\Sigma_2[k],
	\end{multline}
	where $\Sigma_1[k]=M_4\beta_4^k\|\Xi_F[0]\|_2$ and $\Sigma_2[k]=M_4M_5\|R_6\|_2\sum_{l=0}^{k-1}\beta_4^{k-l-1}\beta_5^{l}$.
	It is calculated that $\lim_{k\to\infty}\Sigma_1[k]=0$ and
	\begin{equation}
	\lim_{k\to\infty}\Sigma_2[k]=
	\begin{cases}
	M_4M_5\|R_6\|_2\lim\limits_{k\to\infty}\frac{\beta_5^k-\beta_4^k}{\beta_5-\beta_4}=0,& \text{if}\ \beta_4\neq \beta_5\\
	M_4M_5\|R_6\|_2\lim\limits_{k\to\infty}k\beta_4^{k-1}=0,& \text{if}\ \beta_4=\beta_5
	\end{cases}.\nonumber
	\end{equation}
	Hence, $\lim_{k\to\infty}\|\hat \Xi_F[k]\|_2=0$. By Lemma \ref{lem:Laplacian}, the containment problem is solved.
\end{proof}

\subsection{Applications: Coordinated Control of a Group of Mobile Robots}\label{sec:app}
In order to demonstrate the practical value of the proposed algorithm, this section provides an application example: the coordinated control of a group of mobile robots.

\
Consider the scenario shown in Fig. \ref{fig:back}, where a group of mobile robots are required to move across a partially unknown area through a narrow safe tunnel.
There are three master robots and three slave robots.
As claimed in Introduction Section, the master robots are capable of self-navigation; and the slave robots can measure the relative positions with neighbor robots.

\color{black}
Each robot is a differential drive mobile robot. The schematic diagram of the $i$th mobile robot is given in Fig. \ref{fig:vehicle}.
The coordinate of the center point between two driving wheels is denoted by $(x_i,y_i)$.
The coordinate of the center of the $i$th mobile robot is denoted by $(x_i^c,y_i^c)$.
The $i$th mobile robot's kinematics is described by
\begin{equation}\label{eq:15}
  \begin{cases}
    \dot x_i(t)=v_i(t)\cos(\theta_i)\\
    \dot y_i(t)=v_i(t)\sin(\theta_i)\\
    \dot \theta_i(t)=\omega_i(t)
  \end{cases},
\end{equation}
where $\theta_i$ is the orientation of the $i$th mobile robot with respect to the horizontal axis; $v_i$ and $\omega_i$ are the linear velocity and the angular velocity of the $i$th robot, respectively.

\begin{figure}
  \centering
  %\psfrag{a}{$\theta_i$}
  \includegraphics[width=0.9\hsize]{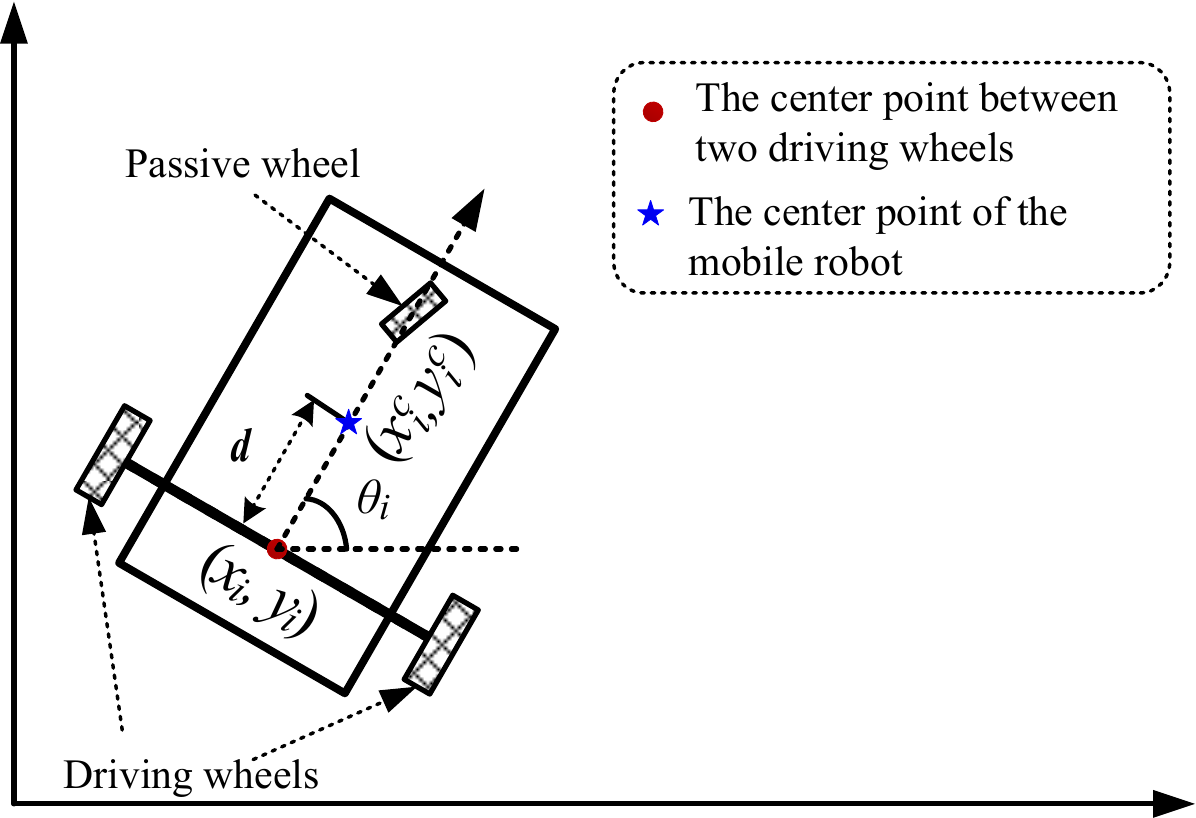}\\
  \caption{The schematic diagram of the $i$th nonholonomic mobile robot.}\label{fig:vehicle}
\end{figure}

Unfortunately, the proposed algorithm cannot be directly applied to this kind of mobile robots due to its nonlinear kinematics.
\
To deal with this challenge, we reformulate \eqref{eq:15} at the point $(x_i^c,y_i^c)$ by the feedback linearization \cite{Ren08RAS}:
\begin{equation}
\dot \varphi_i(t)=u_i(t),\nonumber
\end{equation}
where $\varphi_i(t)=(x_i^c(t),y_i^c(t))^T$, $u_i(t)=(u_i^x(t),u_i^y(t))$,  $u_i^x(t)=v_i(t)\cos(\theta_i)-d\omega_i(t)\sin(\theta_i)$ and $u_i^y(t)=v_i(t)\sin(\theta_i)+d\omega_i(t)\cos(\theta_i)$.
In the rest of this paper, the mobile robot is simply denoted by the coordination of its center point (e.g. the $i$th robot is denoted by $\varphi_i(t)$).
\color{black}

\begin{remark}
According to the above analysis, the control inputs of the $i$th mobile robot are $v_i(t)$ and $w_i(t)$, which can be easily obtained from $u_i(t)$ by using the following transformation
  \begin{equation}
    \begin{pmatrix}
      v_i(t)\\\omega_i(t)
    \end{pmatrix}=
    \begin{pmatrix}
      \cos(\theta_i) & \sin(\theta_i)\\
      -\frac{1}{d}\sin(\theta_i) & \frac{1}{d}\cos(\theta_i)
    \end{pmatrix}u_i(t).\nonumber
  \end{equation}
\end{remark}

Due to the  wide use of digital devices, it is assumed that only the sampled data at each sampling instant is available.
Assume that the sampling period is  {$T=1$}.
For any signal $s(t)$, the sampled data $s(kT)$ at the $k$th sampling instance is denoted by $s[k]$.
By adopting the zero-order holder strategy, the $i$th robot's behavior is described by the following discrete-time difference equation
\begin{equation}\label{eq:13}
  \varphi_i[k+1]=\varphi_i[k]+u_i[k].
\end{equation}
\
Moreover, in this application, it is assumed that the robot's dynamics \eqref{eq:13} is disturbed by the polynomial disturbance $\delta_i[k]=\textbf{1}_2\otimes(1+0.2k-0.01k^2+0.001k^3)$.
The corresponding dynamics of the $i$th robot can therefore be written as
\begin{equation}
\varphi_i[k+1]=\varphi_i[k]+u_i[k]+\delta_i[k].\nonumber
\end{equation}

The task shown in Fig. \ref{fig:back} can be accomplished by adopting the control strategy introduced in Introduction Section.
For each master robot, we select six reference points inside the safe tunnel.
The coordinates of these reference points are shown in Table \ref{tab:r1}.
By the polynomial interpolation, we can obtain the reference trajectory which goes though these points.
The trajectory of the $i$th master robot  is
\begin{equation}\label{eq:r3}
\varphi_i[k]={\bf a}_5^ik^5+{\bf a}_4^ik^4+{\bf a}_3^ik^3+{\bf a}_2^ik^2+{\bf a}_1^ik+{\bf a}_0^i.
\end{equation}
The coefficients ${\bf a}_j^i$ in \eqref{eq:r3} are given in Table \ref{tab:r2}.
\begin{table}
	\centering\caption{The coordinate values of all refrence points.}\label{tab:r1}
	\begin{tabular}{c|c c c}
		\hline\hline
		\backslashbox[1cm]{Time ($k$)}{Leader} & 1 & 2 & 3\\ \hline
		0 & $(0,25)$ & $(20,-5)$ & $(-10,-20)$ \\
		30 & $(110,8)$ & $(130,-15)$ & $(100,-35)$ \\
		60 & $(200,50)$ & $(230,35)$ & $(210,0)$ \\
		90 & $(300,120)$ & $(335,100)$  & $(315,70)$\\
		120 & $(405,155)$ & $(440,130)$ & $(410,110)$\\
		150 & $(475,150)$ & $(510,130)$ & $(480,110)$\\
		\hline\hline
	\end{tabular}
\end{table}
Because the master robots are capable of self-navigation, the master robots can autonomously move along their reference trajectories.

\begin{table}
	\centering
	\caption{Values of coefficients $\bf{a}^i_j$ in \eqref{eq:r3}.}\label{tab:r2}
	\begin{tabular}{c|ccc}
		\hline\hline
		Leader & 1 & 2 & 3\\\hline
		$a_0^i$ & $\begin{pmatrix} 0\\25 \end{pmatrix}$ & $\begin{pmatrix} 20.00\\-5.000 \end{pmatrix}$ & $\begin{pmatrix}-10.00 \\-20.00 \end{pmatrix}$ \\
		$a_1^i$ & $\begin{pmatrix} 4.625\\-1.028 \end{pmatrix}$ & $\begin{pmatrix} 4.100\\-1.392 \end{pmatrix}$ & $\begin{pmatrix} 3.544\\-0.3833 \end{pmatrix}$\\
		$a_2^i\times 10^2$ & $\begin{pmatrix} -4.560 \\-0.7963\end{pmatrix}$ &  $\begin{pmatrix} -1.944\\ 2.801\end{pmatrix}$  & $\begin{pmatrix} 0.7407\\-3.796 \end{pmatrix}$\\
		$a_3^i\times 10^4$ & $\begin{pmatrix}5.092 \\10.77 \end{pmatrix}$ & $\begin{pmatrix} 1.698 \\4.167 \end{pmatrix}$ & $\begin{pmatrix} -1.389\\15.05 \end{pmatrix}$ \\
		$a_4^i\times 10^6$ & $\begin{pmatrix} 1.800 \\-10.91 \end{pmatrix}$ & $\begin{pmatrix} 0\\-6.430 \end{pmatrix}$ & $\begin{pmatrix} 1.029\\-1.337 \end{pmatrix}$\\
		$a_5^i\times 10^8$ & $\begin{pmatrix} 0\\3.086 \end{pmatrix}$ & $\begin{pmatrix} -0.3429\\2.058 \end{pmatrix}$ & $\begin{pmatrix} -0.3429\\3.601 \end{pmatrix}$\\\hline\hline
	\end{tabular}
\end{table}		

\begin{figure}
	\centering
	\includegraphics[width=0.7\hsize]{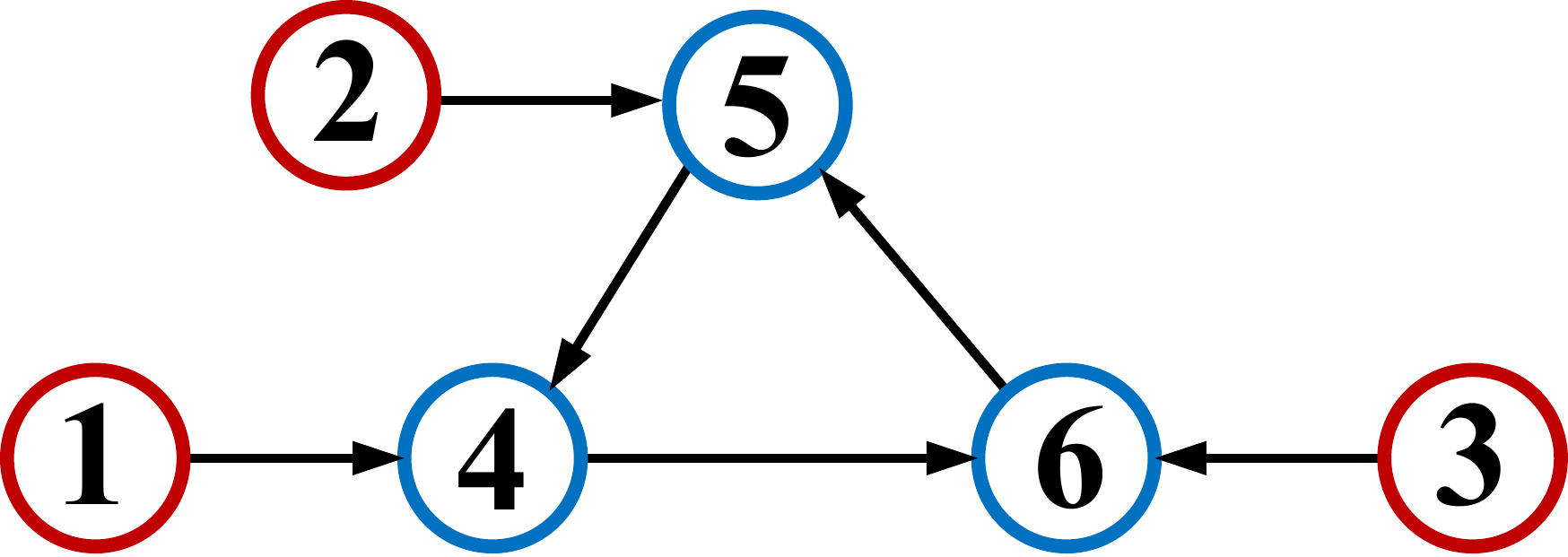}\\
	\caption{Interaction topology: nodes 1, 2 and 3 denote master robots; nodes 4, 5 and 6 denote slave robots.}\label{fig:app_top}
\end{figure}
The communication topology of the multi-robot system is shown in Fig. \ref{fig:app_top}.
If there is a directed edge from node $j$ to node $i$, then the $i$th robot can measure the relative position $e_{ji}[k]$ between itself and robot $j$.
By Section \ref{sec:noise}, the control input of the $(i-3)$th slave robot ($i=4,5,6$) is designed as:
\begin{equation}\label{eq:22}
	u_i[k]=\frac{1}{(1+d_i)}\sum_{l=0}^{5}\kappa_l\D^{-l}\sum_{j=1}^{6}\alpha_{ij}(e_{ji}[k]+\rho_{ji}\eta_{ji}[k]),
\end{equation}
where $\rho_{ji}\eta_{ji}[k]$ is the measurement noise; $\alpha_{ij}=1$ if there is a directed edge from robot $j$ to robot $i$, otherwise $\alpha_{ij}=0$.
By Theorem \ref{thm:dis_1}, the parameters $(k_6,k_5,\cdots,k_0)$ are set to be $(1.806, 0.4769, 0.0786, 0.0085, 5.660\times 10^{-4},1.826\times 10^{-5})$.

In order to verify the effectiveness of the algorithm \eqref{eq:22}, a simulation is carried out.
The simulation results are shown in Fig. \ref{fig:r2}.
Despite of the existence of the disturbance and the measurement noise, all slave robots are convergent into the convex hull spanned by the master robots and move along with them.
%Therefore, the algorithm \eqref{eq:22} solves the coordinated control problem of the multi-robot system.
\begin{figure}`
	\centering
	\includegraphics[width=0.5\textwidth]{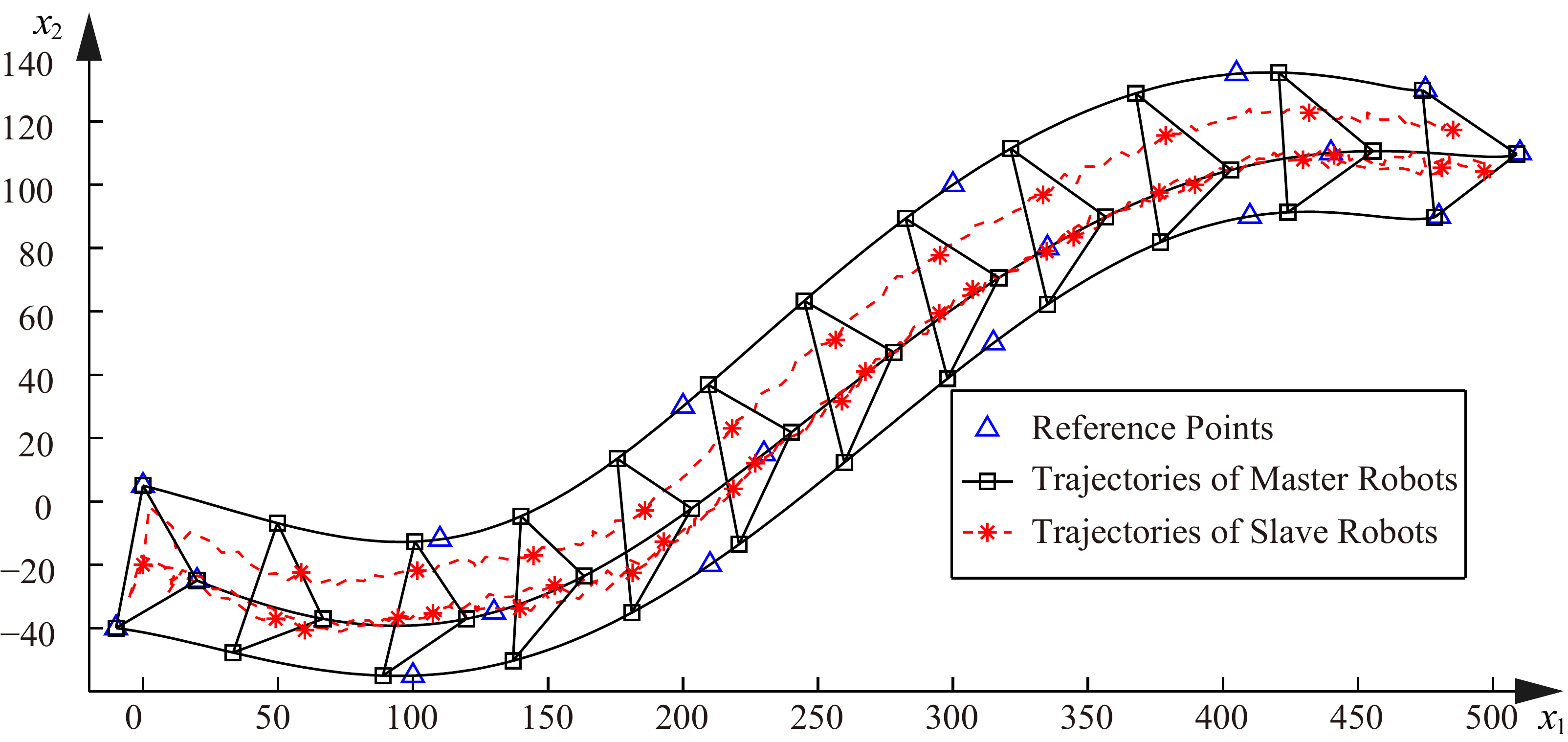}
	\caption{Moving profiles of all robots.}\label{fig:r2}
\end{figure}

\color{black}
\section{Conclusions}\label{sec:conclude}
A $PI^n$-type containment algorithm is proposed for solving the containment problem of MASs in both the continuous-time domain and the discrete-time domain. Leaders in the MAS are assumed to be the polynomial trajectories.
Followers are described by the single-integrator dynamics and the high-order integral dynamics. It is proved that the proposed algorithm can solve the containment problem if for each follower, there is at least one leader which has a directed path to this follower. Compared to the previous results, (1) the containment problem is studied not only in the continuous-time domain but also in the discrete-time domain; (2) the proposed algorithm can solve the containment problem with dynamic leaders even if followers are described by the single-integrator dynamics; (3) there is no non-smooth ``sign'' function in the proposed algorithm; and  {(4) effects of both the disturbance and the measurement noise are taken into account}.
A potential application, the containment control of networked multiple mobile robots, is presented to demonstrate the practical value of the proposed algorithm. 
%In addition, two simulation examples are given, which verify the effectiveness of the proposed algorithm.
%
%From \cite{Andreasson14TAC,Burbano15TAC}, it is proved that the first-order integral term can attenuate the constant disturbance when studying the consensus of MASs. If we apply the proposed protocol to this scenario, due to the existence of high-order integral terms, some kinds of time-varying disturbances may also be dealt with by the proposed protocol. This idea is to be studied in the near future.

\
\appendices
\section{}
\begin{lemma}\label{lem:r1}
	For any sequence of numbers $\{x[k]|k=0,1,2,\cdots\}$, the following formula holds
	\begin{equation}\label{eq:r2}
	\D^n x[k]=\sum_{i=0}^n(-1)^i\binom{n}{i}x[k+n-i],\ n=1,2,\cdots.
	\end{equation}
\end{lemma}
\begin{proof}
	If $n=1$, the correctness of \eqref{eq:r2} can be easily verified.
	Assume when $n=m$, \eqref{eq:r2} is correct.
	It can be calculated that
	\begin{align*}
	&\D^{m+1}x[k]=\D(\D^m x[k])=\D\sum_{i=0}^m(-1)^i\binom{m}{i}x[k+m-i]\\
	=&\sum_{i=0}^m(-1)^i\binom{m}{i}\left(x[k+m-i+1]-x[k+m-i]\right)\\
	=&\sum_{i=0}^{m}(-1)^{i}\binom{m}{i}x[k+m-i+1]\\
	&+\sum_{i=1}^{m+1}(-1)^{i}\binom{m}{i-1}x[k+m-i+1]\\
	=&\sum_{i=1}^{m}(-1)^{i+1}\left(\binom{m}{i}+\binom{m}{i-1}\right)x[k+m-i+1]\\
	&+x[k+m+1]+(-1)^{m+1}x[k]\\
	=&\sum_{i=1}^{m}(-1)^{i+1}\binom{m+1}{i}x[k+m-i+1]\\
	&+x[k+m+1]+(-1)^{m+1}x[k]\\
	=&\sum_{i=0}^{m+1}(-1)^i\binom{m+1}{i}x[k+m+1-i].
	\end{align*}
	By the mathematical induction, it is proved that \eqref{eq:r2} holds.
\end{proof}

\begin{lemma}\label{lem:r2}
	Define a random vector $\nu[k]=(\eta[k],\D\eta[k],\cdots,\D^n\eta[k])^T$, where $\eta[k]$ is the standard white noise.
	For $\forall m\in\{0,1,2\cdots,n\}$, there exists a finite positive constant $G$ such that
	\begin{equation*}
	\|\E(\nu[k]\nu^T[k+m])\|_F<G.
	\end{equation*}
	For $\forall m\ge n+1$, the following equation holds
	\begin{equation*}
	\E(\nu[k]\nu^T[k+m])=\textbf{0}_{(n+1)\times (n+1)}.
	\end{equation*}
\end{lemma}
\begin{proof}
	The $i$th row and the $j$th column entry of $\nu[k]\nu^T[k+m]$ is $\D^{i-1}\eta[k]\D^{j-1}\eta[k+m]$.
	By Lemma \ref{lem:r1}, we know $\D^{i-1}\eta[k]\D^{j-1}\eta[k+m]$ is a linear combination of the following terms
	\begin{equation*}
	\eta[k+i-1-s]\eta[k+m+j-1-t],
	\end{equation*}
	where $s\in\{0,1,\cdots,i-1\}$ and $t\in\{0,1,\cdots,j-1\}$.
	
	If $m\in\{0,1,2\cdots,n-1\}$, then $\forall s\in\{0,1,\cdots,i-1\}$ and $\forall t\in\{0,1,\cdots,j-1\}$,
	\begin{multline*}
	|\E(\eta[k+i-1-s]\eta[k+m+j-1-t])|_2=\\
	\begin{cases}
	0,\ m+j-i+s-t\neq 0 \\
	1,\ m+j-i+s-t= 0
	\end{cases}.
	\end{multline*}
	Hence, there must exist a finite positive constant $G_1$ such that
	\begin{equation*}
	|\E(\D^{i-1}\eta[k]\D^{j-1}\eta[k+m])|_2<G_1,
	\end{equation*}
	which follows that
	\begin{equation*}
	\|\E(\nu[k]\nu^T[k+m])\|_F<\sqrt{(n+1)^2G_1^2}=(n+1)G_1.
	\end{equation*}
	Let $G=(n+1)G_1<\infty$. Then $\|\E(\nu[k]\nu^T[k+m])\|_2<G$.
	
	\hspace{0.25cm} If $m\ge n+1$, then $k+m+j-1-t\ge k+n+1+j-1-(j-1)=k+n+1$ and $k+i-1-s\le k+n+1-1-0=k+n$. Hence, $\E(\eta[k+i-1-s]\eta[k+m+j-1-t])=0$, which leads to that $\E(\D^{i-1}\eta[k]\D^{j-1}\eta[k+m])=0$.
	Therefore $\E(\nu[k]\nu^T[k+m])=\textbf{0}_{(n+1)\times (n+1)}$.
\end{proof}

\begin{lemma}\label{lem:r3}
	The standard white noise $\eta[k]$ has the following property:
	\begin{equation}
	\E(\D^l\eta[k])=0,\quad l=0,1,2,\cdots.
	\end{equation}
\end{lemma}
\begin{proof}
	It is easy to see $\E(\D^0\eta[k])=E(\eta[k])=0$.
	Assume that $\E(\D^j\eta[k])=0$. Then it can be obtained that
	\begin{multline*}
	\E(\D^{j+1}\eta[k])=\E(\D^j\eta[k+1]-\D^j\eta[k])\\
	=\E(\D^j\eta[k+1])-\E(\D^j\eta[k])=0.
	\end{multline*}
	By the mathematic induction, it is proved that $\E(\D^l\eta[k])=0, l=0,1,2,\cdots$.
\end{proof}

\color{black}

\IEEEpeerreviewmaketitle

\ifCLASSOPTIONcaptionsoff
  \newpage
\fi

\balance

\bibliographystyle{IEEEtran}
\bibliography{reference}

\begin{thebibliography}{10}
\providecommand{\url}[1]{#1}
\csname url@samestyle\endcsname
\providecommand{\newblock}{\relax}
\providecommand{\bibinfo}[2]{#2}
\providecommand{\BIBentrySTDinterwordspacing}{\spaceskip=0pt\relax}
\providecommand{\BIBentryALTinterwordstretchfactor}{4}
\providecommand{\BIBentryALTinterwordspacing}{\spaceskip=\fontdimen2\font plus
\BIBentryALTinterwordstretchfactor\fontdimen3\font minus
  \fontdimen4\font\relax}
\providecommand{\BIBforeignlanguage}[2]{{%
\expandafter\ifx\csname l@#1\endcsname\relax
\typeout{** WARNING: IEEEtran.bst: No hyphenation pattern has been}%
\typeout{** loaded for the language `#1'. Using the pattern for}%
\typeout{** the default language instead.}%
\else
\language=\csname l@#1\endcsname
\fi
#2}}
\providecommand{\BIBdecl}{\relax}
\BIBdecl

\bibitem{Hu15TCYB}
W.~Hu, L.~Liu, and G.~Feng, ``Consensus of linear multi-agent systems by
  distributed event-triggered strategy,'' \emph{IEEE Transactions on
  Cybernetics}, in press, DOI: 10.1109/TCYB.2015.2398892, 2015.

\bibitem{Li15TCYB}
S.~Li, G.~Feng, X.~Luo, and X.~Guan, ``Output consensus of heterogeneous linear
  discrete-time multiagent systems with structural uncertainties,'' \emph{IEEE
  Transactions on Cybernetics}, in press, DOI: 10.1109/TCYB.2015.2388538, 2015.

\bibitem{Hou09SMCB}
Z.-G. Hou, L.~Cheng, and M.~Tan, ``Decentralized robust adaptive control for
  the multiagent system consensus problem using neural networks,'' \emph{IEEE
  Transactions on Systems, Man, and Cybernetics, Part B: Cybernetics}, vol.~3,
  no.~39, pp. 636--647, 2009.

\bibitem{Cheng10TNN}
L.~Cheng, Z.-G. Hou, M.~Tan, Y.~Lin, and W.~Zhang, ``Neural-network-based
  adaptive leader-following control for multi-agent systems with
  uncertainties,'' \emph{IEEE Transactions on Neural Networks}, vol.~8, no.~21,
  pp. 1351--1358, 2010.

\bibitem{Vaugha00RAS}
R.~Vaughan, N.~Sumpter, J.~Henderson, A.~Frost, and S.~Cameron, ``Experiments
  in automatic flock control,'' \emph{Robotics and Autonomous Systems},
  vol.~30, no. 1--2, pp. 109--117, 2000.

\bibitem{Haque08ACC}
M.~A. Haque, M.~Egerstedt, and C.~F. Martin, ``First-order, networked control
  models of swarming silkworm moths,'' in \emph{Proceedings of American Control
  Conference}, Seattle, USA, 2008, pp. 3798--3803.

\bibitem{Galbusera13SCL}
L.~Galbusera, G.~Ferrari-Trecate, and R.~Scattolini, ``A hybrid model
  predictive control scheme for containment and distributed sensing in
  multi-agent systems,'' \emph{Systems \& Control Letters}, vol.~62, no.~5, pp.
  413--419, 2013.

\bibitem{Parker03RS}
L.~E. Parker, B.~Kannan, X.~Fu, and Y.~Tang, ``Heterogeneous mobile sensor net
  deployment using robot herding and line-of-sight formations,'' in
  \emph{Proceedings of the 2003 IEEE/RSJ International Conference on
  Intelligent Robots and Systems}, Las Vegas, USA, 2003, pp. 2488--2493.

\bibitem{Cao11TCST}
Y.~Cao, D.~Stuart, W.~Ren, and Z.~Meng, ``Distributed containment control for
  multiple autonomous vehicles with double-integrator dynamics: Algorithms and
  experiments,'' \emph{IEEE Transactions on Control Systems Technology},
  vol.~19, no.~4, pp. 929--938, 2011.

\bibitem{Wang15CCC}
Y.~Wang, L.~Cheng, Z.-G. Hou, M.~Tan, and H.~Yu, ``Coordinated transportation
  of a group of unmanned ground vehicles,'' in \emph{Proceedings of The 34th
  Chinese Control Conference and SICE Annual Conference 2015}, Hangzhou, China,
  2015, pp. 7027--7032.

\bibitem{Dimarogonas06CDC}
D.~Dimarogonas, M.~Egerstedt, and K.~Kyriakopoulos, ``A leader-based
  containment control strategy for multiple unicycles,'' in \emph{Proceedings
  of the 45th IEEE Conference on Decision and Control}, San Diego, USA, 2006,
  pp. 5968--5973.

\bibitem{Ferrari06WHS}
G.~Ferrari-Trecate, M.~Egerstedt, A.~Buffa, and M.~Ji, ``Laplacian sheep: A
  hybrid, stop-go policy for leader-based containment control,'' in
  \emph{Proceedings of the 9th International Workshop on Hybrid Systems:
  Computation and Control}, Santa Barbara, USA, 2006, pp. 212--226.

\bibitem{Ji08TAC}
M.~Ji, G.~Ferrari-Trecate, M.~Egerstedt, and A.~Buffa, ``Containment control in
  mobile networks,'' \emph{IEEE Transactions on Automatic Control}, vol.~53,
  no.~8, pp. 1972--1975, 2008.

\bibitem{Notarstefano11Automatica}
G.~Notarstefano, M.~Egerstedt, and M.~Haque, ``Containment in leader-follower
  networks with switching communication topologies,'' \emph{Automatica},
  vol.~47, no.~5, pp. 1035--1040, 2011.

\bibitem{Cao12Automatica}
Y.~Cao, W.~Ren, and M.~Egerstedt, ``Distributed containment control with
  multiple stationary or dynamic leaders in fixed and switching directed
  networks,'' \emph{Automatica}, vol.~48, no.~8, pp. 1586--1597, 2012.

\bibitem{Tang12AAA}
Z.-J. Tang, T.-Z. Huang, and J.-L. Shao, ``Containment control of multiagent
  systems with multiple leaders and noisy measurements,'' \emph{Abstract and
  Applied Analysis}, vol. 2012, Article ID 262153, 2012.

\bibitem{Wang14TCYB}
X.~Wang, S.~Li, and P.~Shi, ``Distributed finite-time containment control for
  double-integrator multiagent systems,'' \emph{IEEE Transactions on
  Cybernetics}, vol.~9, no.~44, pp. 1518--1528, 2014.

\bibitem{Liu12Automatica}
H.~Liu, G.~Xie, and L.~Wang, ``Necessary and sufficient conditions for
  containment control of networked multi-agent systems,'' \emph{Automatica},
  vol.~48, no.~7, pp. 1415--1422, 2012.

\bibitem{Li14NAHS}
J.~Li, Z.-H. Guan, R.-Q. Liao, and D.-X. Zhang, ``Impulsive containment control
  for second-order networked multi-agent systems with sampled information,''
  \emph{Nonlinear Analysis: Hybrid Systems}, vol.~12, pp. 93--103, 2014.

\bibitem{Li12TAC}
J.~Li, W.~Ren, and S.~Xu, ``Distributed containment control with multiple
  dynamic leaders for double-integrator dynamics using only position
  measurements,'' \emph{IEEE Transactions on Automatic Control}, vol.~57,
  no.~6, pp. 1553--1559, 2012.

\bibitem{Zhang14SCL}
B.~Zhang, Y.~Jia, and F.~Matsuno, ``Finite-time observers for multi-agent
  systems without velocity measurements and with input saturations,''
  \emph{Systems \& Control Letters}, vol.~68, pp. 86--94, 2014.

\bibitem{Lou12Automatica}
Y.~Lou and Y.~Hong, ``Target containment control of multi-agent systems with
  random switching interconnection topologies,'' \emph{Automatica}, vol.~48,
  no.~5, pp. 879--885, 2012.

\bibitem{Wang14Automatica}
Y.~Wang, L.~Cheng, Z.-G. Hou, M.~Tan, and M.~Wang, ``Containment control of
  multi-agent systems in a noisy communication environment,''
  \emph{Automatica}, vol.~50, no.~7, pp. 1922--1928, 2014.

\bibitem{Liu13SCL}
H.~Liu, G.~Xie, and L.~Wang, ``Containment of linear multi-agent systems under
  general interaction topologies,'' \emph{Systems \& Control Letters}, vol.~61,
  no.~4, pp. 528--534, 2012.

\bibitem{Ma14Neuro}
Q.~Ma and G.~Miao, ``Distributed containment control of linear multi-agent
  systems,'' \emph{Neurocomputing}, vol. 133, pp. 399--403, 2014.

\bibitem{Li13JRNC}
Z.~Li, W.~Ren, X.~Liu, and M.~Fu, ``Distributed containment control of
  multi-agent systems with general linear dynamics in the presence of multiple
  leaders,'' \emph{International Journal of Robust and Nonlinear Control},
  vol.~23, no.~5, pp. 62--75, 2013.

\bibitem{Liu15Automatica}
H.~Liu, L.~Cheng, M.~Tan, and Z.-G. Hou, ``Containment control of
  continuous-time linear multi-agent systems with aperiodic sampling,''
  \emph{Automatica}, vol.~7, no.~57, pp. 78--84, 2015.

\bibitem{Mei12Automatica}
J.~Mei, W.~Ren, and G.~Ma, ``Distributed containment control for {Lagrangian}
  networks with parametric uncertainties under a directed graph,''
  \emph{Automatica}, vol.~48, no.~12, pp. 653--659, 2012.

\bibitem{Yoo14ES}
S.~J. Yoo, ``Distributed adaptive containment control of networked
  flexible-joint robots using neural networks,'' \emph{Expert Systems with
  Applications}, vol.~41, no.~2, pp. 470--477, 2014.

\bibitem{Meng10Automatica}
Z.~Meng, W.~Ren, and Z.~You, ``Distributed finite-time attitude containment
  control for multiple rigid bodies,'' \emph{Automatica}, vol.~46, no.~12, pp.
  2092--2099, 2010.

\bibitem{XKWang14Cyb}
X.~Wang, J.~Qin, and C.~Yu, ``Iss method for coordination control of nonlinear
  dynamical agents under directed topology,'' \emph{IEEE Transactions on
  Cybernetics}, vol.~44, no.~10, pp. 1832--1845, 2014.

\bibitem{Andreasson14TAC}
M.~Andreasson, D.~V. Dimarogonas, and H.~Sandberg, ``Distributed control of
  networked dynamical systems: static feedback, integral action and
  consensus,'' \emph{IEEE Transactions on Automatic Control}, vol.~59, no.~7,
  pp. 1750--1764, 2014.

\bibitem{Burbano15TAC}
D.~Burbano and M.~di~Bernardo, ``Distributed {PID} control for consensus of
  homogeneous and heterogeneous networks,'' \emph{IEEE Transactions on Control
  of Network Systems}, vol.~2, no.~2, pp. 154--163, 2015.

\bibitem{Demetriou14SCL}
M.~A. Demetriou, ``Spatial pid consensus controllers for distributed filters of
  distributed parameter systems,'' \emph{Systems \& Control Letters}, vol.~1,
  no.~63, pp. 57--62, 2014.

\bibitem{Wang14Cyb}
X.~Wang, S.~Li, and P.~Shi, ``Distributed finite-time containment control for
  double-integrator multiagent systems,'' \emph{IEEE Transactions on
  Cybernetics}, vol.~44, no.~9, pp. 1518--1528, 2014.

\bibitem{Meng11SMCB}
Z.~Meng, W.~Ren, Y.~Cao, and Z.~You, ``Leaderless and leader-following
  consensus with communication and input delays under a directed network
  topology,'' \emph{IEEE Transactions on Systems, Man, and Cybernetics--Part B:
  Cybernetics}, vol.~41, no.~1, pp. 75--88, 2011.

\bibitem{Li11IJRNC}
Z.~Li, W.~Ren, X.~Liu, and M.~Fu, ``Distributed containment control of
  multi-agent systems with general linear dynamics in the presence of multiple
  leaders,'' \emph{International Journal of Robust and Nonlinear Control},
  vol.~23, no.~5, pp. 534--547, 2013.

\bibitem{Wen13IJC}
G.~Wen, Z.~Li, Z.~Duan, and G.~Chen, ``Distributed consensus control for linear
  multi-agent systems with discontinuous observations,'' \emph{International
  Journal of Control}, vol.~86, no.~1, pp. 95--106, 2013.

\bibitem{Kristian13Automatica}
K.~Hengster-Movrica, K.~You, F.~L. Lewis, and L.~Xie, ``Synchronization of
  discrete-time multi-agent systems on graphs using {Riccati} design,''
  \emph{Automatica}, vol.~49, no.~2, pp. 414--423, 2013.

\bibitem{Ren08RAS}
W.~Ren and N.~Sorensen, ``Distributed coordination architecture for multi-robot
  formation control,'' \emph{Robotics and Autonomous Systems}, vol.~56, no.~4,
  pp. 324--333, 2008.

\end{thebibliography}

\end{document}